\documentclass[pdftex,11pt]{article}%
\usepackage[pdftex]{hyperref}%
\usepackage[pdftex]{graphicx}%
\usepackage[normalem]{ulem}
\usepackage{amsmath,amsthm,amsfonts,
amsbsy,amssymb,upref,enumerate,bigstrut,
color,mathtools,mathrsfs,float,bm,dsfont,scalefnt}
\usepackage{lipsum}
\usepackage{stmaryrd}
\usepackage{authblk} % author and affiliations

\usepackage[left=1in,top=1.2in,right=0.7in]{geometry}

\def\orcid#1{\kern .08em\href{https://orcid.org/#1}{\includegraphics[keepaspectratio,width=0.7em]{parametros/orcid.pdf}}}

\theoremstyle{definition}
\newtheorem{remark1}{Remark}

\newtheorem{theorem}{Theorem}[section]

\newtheorem{corollary}[theorem]{Corollary}

\newtheorem{lemma}[theorem]{Lemma}
\newtheorem{proposition}[theorem]{Proposition}

\numberwithin{equation}{section} % requires amsmath

\makeatletter
\def\@seccntformat#1{\@ifundefined{#1@cntformat}%
	{\csname the#1\endcsname\quad}%      default
	{\csname #1@cntformat\endcsname}%    enable individual control
}
\makeatother

\newif\ifShowComments
\ShowCommentstrue
\def\strutdepth{\dp\strutbox}
\def\druk#1{\strut\vadjust{\kern-\strutdepth
        {\vtop to \strutdepth{%
                \baselineskip\strutdepth\vss
                        \llap{\hbox{#1}\quad}\null}}}}

\newcommand\scalemath[2]{\scalebox{#1}{\mbox{\ensuremath{\displaystyle #2}}}} % reduz tamanho de matriz e formulas mat

\newtheorem{thm1}{Theorem}

\title{\bf
On the bimodal Gumbel model with application to environmental data
}
\author[1]{Cira E. G. Otiniano \thanks{ciragotiniano@gmail.com}}
\author[1]{Roberto Vila \thanks{rovig161@gmail.com}}
\author[1]{Pedro C.  Brom  \thanks{pcbrom@gmail.com} }
\author[2]{Marcelo Bourguignon$^*$ \thanks{m.p.bourguignon@gmail.com}}
\affil[1]{Departamento de Estat\'istica, Universidade de Bras\'ilia, 70910-900, Bras\'ilia, Brazil}
\affil[2]{Departamento de Estat\'istica, Universidade Federal do Rio Grande do Norte, 59078-970, Natal/RN, Brazil}
\begin{document}
\maketitle

\begin{abstract}
{

The Gumbel model is a very popular statistical model due to its wide applicability for instance in the course of certain survival, environmental, financial or reliability studies.
In this work, we have introduced a bimodal generalization of the Gumbel distribution that
can be an alternative to model bimodal data.
We derive the analytical shapes of the corresponding probability density function and the
hazard rate function and provide graphical illustrations.
Furthermore, We have discussed the properties of this density such as mode, bimodality, moment generating function and moments.
Our results were verified using the Markov chain Monte Carlo simulation method.
The maximum likelihood method is used for parameters estimation.
Finally, we also carry out an application to real data that demonstrates the usefulness of the proposed distribution.
}

\end{abstract}

\tableofcontents

\section{Introduction}
%\label{sec:1}
\noindent

Asymmetrical models for a real-valued random variable such as the Gumbel and generalized extreme value distributions have been extensively utilized for modeling various random phenomena encountered for instance in the course of certain environmental, financial or reliability studies.

Let $Y_1, Y_2, \dots, Y_n$ be a series of independent random
variables with common distribution function $F$, and
$M_n= \max\{ Y_1, Y_2, \dots, Y_n \}$. The Gumbel distribution is one of the extreme value distributions, characterized by Fisher and Tippet (1928) \cite{FT28} as the limit distribution for maxima. That is, if exist the normalization sequences $ a_n $ and $ b_n> 0 $ such that $ P[(M_n-a_n) / b_n] $ converges to a non-degenerate distribution $G$, then $G$ is an extreme value distribution. The $G$ distribution must be one of three types: Fréchet,  Gumbel or negative Weibull. Among these three distributions the Gumbel distribution (see Gumbel (1958) \cite{G58} ) is the only one with a light tail. For this reason it can be considered  as a good alternative for modeling extreme data whose tails are not heavy.

A random variable $Y$ have a Gumbel distribution with location parameter $\mu\in\mathbb{R}$ and scale paramer $\sigma>0$,  denoted by
$Y\sim F_G(\cdot;\mu,\sigma)$, if its probability density function (PDF) and cumulated distribution function (CDF) are given, respectively,  by
\begin{eqnarray}\label{pdf1}
f_G(y;\mu,\sigma)
=
\frac{1}{\sigma}\exp\biggl\{-\Big(\frac{y-\mu}{\sigma} \Big)-\exp\Big[-\Big(\frac{y-\mu}{\sigma} \Big)\Big]\biggr\}
\quad
\end{eqnarray}
and
\begin{eqnarray}\label{cdfgev1}
F_G(y;\mu,\sigma)
=
\displaystyle
\exp\biggl\{-\exp\Big[-\Big(\frac{y-\mu}{\sigma} \Big)
\Big]  \biggr\}, \quad y\in \mathbb{R}.
\end{eqnarray}

% This distribution was derived and analyzed by
%Gumbel in its book published  im 1958.
%This book provides a theoretical foundation for statisticians, engineers and scientists. On theory and applications to different areas, we can mention %the books by Embrechts et al. (1997), Reiss (2007).

Generalizations of the Gumbel distribution have been proposed by several authors. Pinheiro and Ferrari (2015) \cite{Pinheiro15} carried out a comprehensive review of the generalization of Gumbel distribution and after comparing them they   conclude that some distributions suffer  from overparameterization. Another generalization of Gumbel is the Exponentiated Gumbel Type-2 by  Okorie (2016) \cite{Okorie2016} and the references in  \cite{Okorie2017} and \cite{Pinheiro15}.

Despite its broad applicability in many fields the Gumbel is not suitable to model
bimodal data. Furthermore, all the models cited above are not suitable for capturing this.
In this context, in this paper, we propose the bimodal Gumbel (BG)  distribution as an alternative model of extreme data with more than one mode.
Our approach consists of introduce bimodality in \eqref{pdf1} as Elal-Olivero (2010) \cite{EL10}.  Results involving bimodality
in other related probabilistic models  can be found, for example,
in	Martinez et al. (2013) \cite{MAJS13}, Çankaya et al. (2015) \cite{CBDA15};
and more recently, in Vila et al. (2020-2021) \cite{VLSSS19,VFSPO20,VC20, VHR21}.
The advantage of our model in comparison to other generalizations of the Gumbel distribution is the number of parameters and the fact that it can be used to model extreme data with one or two modes.

This paper is organized as follows. In Section \ref{Sec:2}, we define the BG distribution.
In Section \ref{Sec:3}, we provide general properties of the BG distribution including the cumulative distribution function, hazard rate function, mode, bimodality, moment generating function and moments, and stochastic representation.
In Section \ref{Sec:4}, we provide graphical illustrations.
Estimation of the parameters by maximum likelihood is investigated in Section \ref{Sec:5}.
In Section \ref{Sec:6}, we discuss an application to real data.
Some conclusions are addressed in Section \ref{Sec:7}.

\section{The bimodal Gumbel model}\label{Sec:2}
\noindent
%We introduce bimodality in \eqref{pdf1}
%through  of a quadratic transformation. For this, consider a random variable $Y$ with PDF
%$F_{\rm G}(x;\mu,\sigma)$, as \eqref{cdfgev1}, then
%$\mathbb{E}(Y)=\mu+\sigma\gamma$   and
%$\mathrm{Var}(Y)=\sigma^2{\pi^2\over 6}$.
%Then
We said that a real-valued random variable $X$ has a bimodal Gumbel (BG) distribution with parameters $ \mu\in\mathbb{R}, \sigma>0$ and $\delta\in\mathbb{R}$,
denoted by $X\sim F_{\rm BG}(\cdot;\mu,\sigma, \delta)$, if its PDF is given by
%
%%%%%%%%%%%%%%%%%%%%%%%%%
\begin{eqnarray}\label{pdfbgev2}
	f_{\rm BG}(x;\mu,\sigma, \delta)
%	&=\displaystyle
%	{1 \over Z_{\delta}}\,
%	\big[(1-\delta x)^2+1\big] \,  f_{\rm G}(x;\mu,\sigma),  \quad x\in \mathbb{R}\nonumber
%	\\[0,2cm]
	=
	{ \displaystyle
	\big[(1-\delta x)^2+1\big]
	\exp\biggl\{-\Big(\frac{x-\mu}{\sigma} \Big)-\exp\Big[-\Big(\frac{x-\mu}{\sigma} \Big)\Big]\biggr\}
	\over \textstyle
	\sigma\big[
	1+\delta^2\sigma^2{\pi^2\over 6}+(\delta\mu+\delta\sigma\gamma-1)^2
	\big]
}
,  \quad x\in \mathbb{R},
\end{eqnarray}
%%%%%%%%%%%%%%%%%%
where $\gamma$  is the Euler's constant.

Let
\begin{eqnarray}\label{partition}
	Z_{\delta}
	=
	1+\delta^2\sigma^2{\pi^2\over 6}+(\delta\mu+\delta\sigma\gamma-1)^2
\end{eqnarray}
be the normalization constant of the BG distribution \eqref{pdfbgev2}. Using this notation, note that
$f_{\rm BG}(x;\mu,\sigma, \delta)=	{1 \over Z_{\delta}}\,[(1-\delta x)^2+1] \,  f_{\rm G}(x;\mu,\sigma)$ and that $f_{\rm BG}(x;\mu,\sigma, 0)=f_{\rm G}(x;\mu,\sigma)$. In other words, we introduce bimodality in the Gumbel distribution \eqref{pdf1} through  of a quadratic transformation $(1-\delta x)^2+1$.

The BG distribution  function \eqref{pdfbgev2} is well defined, because
\begin{eqnarray*}
	\int_{-\infty}^{\infty}f_{\rm BG}(x;\mu,\sigma, \delta) \,{\rm d}x
	&=&
	{1 \over Z_{\delta}}\,
	[1+\mathbb{E}(1-\delta Y)^2], \quad Y\sim F_G(\cdot;\mu,\sigma)
	\\[0,1cm]
	&=&
	{1 \over Z_{\delta}}\,
	\big\{1+ \delta^2\mathrm{Var}(Y)+[\delta\,\mathbb{E}(Y)-1]^2\big\}
	=1,
\end{eqnarray*}
where $\mathbb{E}(Y)=\mu+\sigma\gamma$   and
$\mathrm{Var}(Y)=\sigma^2{\pi^2\over 6}$.

\smallskip

%In this case,
The CDF of a BG random variable $X\sim F_{\rm BG}(\cdot;\mu,\sigma, \delta)$, defined for  $x\in\mathbb{R}$, is given by (see Subsection \ref{CDF} for more details)
%%%%%%%%%%%%%% CDF Gamma 0 %%%%%%%%%%%%%
%
\begin{align} \label{cdfbgev2-1}
F_{\rm BG}(x; \mu,\sigma, \delta)
=&
{\displaystyle
	\big[2-\delta\mu(2-\delta\mu)\big]
	\exp\biggl\{-\exp\Big[-\Big(\frac{x-\mu}{\sigma} \Big)\Big]\biggr\}
	+
	\delta^2\sigma^2
	I\biggl(2;\exp\Big[-\Big(\frac{x-\mu}{\sigma} \Big)\Big],+\infty\biggr)
	\over
	1+\delta^2\sigma^2{\pi^2\over 6}+(\delta\mu+\delta\sigma\gamma-1)^2
}
\nonumber
\\[0,2cm]
&
+
{\displaystyle
	2\delta(1-\delta\mu)
	\Biggl\{
	{(x-\mu)} \,
	\exp\biggl\{-\exp\Big[-\Big(\frac{x-\mu}{\sigma} \Big)\Big]\biggr\}
	-
	\sigma
	\Gamma\biggl(0,\exp\Big[-\Big(\frac{x-\mu}{\sigma} \Big)\Big]\biggr)\Biggr\}
	\over
	1+\delta^2\sigma^2{\pi^2\over 6}+(\delta\mu+\delta\sigma\gamma-1)^2
},
\end{align}
where $\Gamma(a,b)$ is the upper incomplete gamma function and
\begin{eqnarray}\label{Iab}
I(k;a,b)=(-1)^k \int_{a}^{b} \ln^k(v) \exp(-v) \, {\rm d}v,
\quad k\in\mathbb{N}\cup\{0\}, \ 0\leqslant a<b\leqslant +\infty
\end{eqnarray}
is the incomplete moments of the random variable $V\sim F_G(\cdot;0,1)$.
For $k > 1$, closed form solutions for the definite integral $I(k;a,b)$ are not available in terms of commonly used functions.

From the formula in \eqref{cdfbgev2-1},
when $\delta=0$,
$F_{\rm BG}(x; \mu,\sigma, 0)	
=
\exp\{-\exp[-(\frac{x-\mu}{\sigma} )
]\}
=F_G(x;\mu,\sigma)$, $x\in\mathbb{R}$, what was known in \eqref{cdfgev1}.

\begin{remark1}\label{eq-main}
	Let $X\sim F_{\rm BG}(\cdot;\mu,\sigma, \delta)$ and let $g(\cdot)$ be a real-valued Borel measurable function.
	From definition of expectation  and by using  the PDF of the BG distribution \eqref{pdfbgev2}, we have
	\begin{eqnarray*}
	\mathbb{E}[g(X)]
%	&=&
%	\int_{-\infty}^{\infty}g(x)f_{\rm BG}(x;\mu,\sigma, \delta)\, {\rm d}x
%	\\[0,1cm]
		&=&
	{1 \over Z_{\delta}}\,
	\mathbb{E}\big[g(Y)(1-\delta Y)^2+g(Y)\big]
	\\[0,1cm]
	&=&
	{1 \over Z_{\delta}}\,
	\big\{ 2\mathbb{E}[g(Y)]
	-2\delta \mathbb{E}[Yg(Y)]
	+
	\delta ^2 \mathbb{E}[Y^2g(Y)]\big\},
	\end{eqnarray*}
	where $Y\sim F_G(\cdot;\mu,\sigma)$ and $Z_\delta$ is as in \eqref{partition}.
\end{remark1}

\section{Some properties of the BG distribution} \label{Sec:3}
\noindent

In this section, some mathematical properties as closed expression for the CDF, rate, modes, bimodality, hazard function and moments of the  BG distribution are discussed.

\subsection{Cumulative distribution function}\label{CDF}
\noindent

In this subsection, we derive in detail the closed expression in \eqref{cdfbgev2-1} for the CDF of a BG random variable $X\sim F_{\rm BG}(\cdot;\mu,\sigma, \delta)$.

Indeed, from Remark \ref{eq-main} with $g(X)=\mathds{1}_{X\leqslant x}$, we obtain
\begin{eqnarray}\label{cdfbgev2}
F_{\rm BG}(x; \mu,\sigma, \delta)
=
\frac{1}{Z_{\delta}}
\big[ 2F_{\rm G}(y; \mu, \sigma)
-2\delta \mathbb{E}\left( Y \mathds{1}_{Y\leq x}\right)
+
\delta ^2 \mathbb{E}\left( Y^2 \mathds{1}_{Y\leq x}\right)\big],
\end{eqnarray}
where $Y\sim F_G(\cdot;\mu,\sigma)$ and $Z_{\delta}$ is as in \ref{partition}.

Taking the following change of variables
$
z=\exp[-({y-\mu})/{\sigma} ], \
{\rm d} z= -({z/\sigma}) {\rm d}y,
\ y=\mu-\sigma\ln(z),
$
and using a binomial expansion we have
\begin{eqnarray}\label{form-mom-trunc}
\mathbb{E}\big( Y^k \mathds{1}_{Y\leq x}\big)
&=&
\int_{\exp[-(\frac{x-\mu}{\sigma} )]}^{+\infty}
[\mu-\sigma\ln(z)]^k \exp(-z)\, {\rm d}z
\nonumber
\\[0,2cm]
&=&
\sum_{i=0}^{k} \binom{k}{i} (-1)^i \mu^{k-i}\sigma^i \int_{\exp[-(\frac{x-\mu}{\sigma} )]}^{+\infty}
\ln^i(z) \exp(-z)\, {\rm d}z
\nonumber
\\[0,2cm]
&=&
\sum_{i=0}^{k} \binom{k}{i} \mu^{k-i}\sigma^i \, I\biggl(i;\exp\Big[-\Big(\frac{x-\mu}{\sigma} \Big)\Big],+\infty\biggr),
\end{eqnarray}
where $k\in\mathbb{N}\cup\{0\}$ and $I(i;a,b)$ is as in  \eqref{Iab}. By combining \eqref{cdfbgev2} and \eqref{form-mom-trunc}, and by using the relations
\begin{eqnarray*}
I(0;a,+\infty)&=&\exp(-a),
\\[0,2cm]
I(1;a,+\infty)&=& -\exp(-a)\ln(a) -\Gamma(0,a),
\end{eqnarray*}
we get formula \eqref{cdfbgev2-1}.
%, that is,
%\begin{align}\label{cdf}
%F_{\rm BG}(x; \mu,\sigma, \delta)
%=&
%{\displaystyle
%	\big[2-\delta\mu(2-\delta\mu)\big]
%	\exp\biggl\{-\exp\Big[-\Big(\frac{x-\mu}{\sigma} \Big)\Big]\biggr\}
%	+
%	\delta^2\sigma^2
%	I\biggl(2;\exp\Big[-\Big(\frac{x-\mu}{\sigma} \Big)\Big],+\infty\biggr)
%	\over
%	1+\delta^2\sigma^2{\pi^2\over 6}+(\delta\mu+\delta\sigma\gamma-1)^2
%}
%\nonumber
%\\[0,2cm]
%&
%+
%{\displaystyle
%	2\delta(1-\delta\mu)
%	\Biggl\{
%	{(x-\mu)} \,
%	\exp\biggl\{-\exp\Big[-\Big(\frac{x-\mu}{\sigma} \Big)\Big]\biggr\}
%	-
%	\sigma
%	\Gamma\biggl(0,\exp\Big[-\Big(\frac{x-\mu}{\sigma} \Big)\Big]\biggr)\Biggr\}
%	\over
%	1+\delta^2\sigma^2{\pi^2\over 6}+(\delta\mu+\delta\sigma\gamma-1)^2
%}.
%\end{align}

%%%%%%%%%%%%%%%%%%%%%%%%
%\subsubsection{Monotonicity}\label{Sub-Monotonicity}
%%%%%%%%%%%%%%%%%%%%%%% %%%%%%%%%%%%% Prop  2 %%%%%%%%%%%%%%
%
%\begin{proposition}\label{prop-BG2}
%	The {\rm BG} density $f_{\rm BG}(x;\xi,\mu,\sigma,\delta)$ of the {\rm BG} distribution \eqref{pdfbgev2}, with $\delta\in\mathbb{R}$, is increasing for each $x<\mu$ such that $\delta x>1$.
%\end{proposition}
%\begin{proof}
%	The proof of this proposition follows  the same reasoning as the one of Proposition \ref{mon-1}. Then, it is omitted.
%\end{proof}

%%%%%%%%%%%%%%%%%%%%%%%%%%%%%%%%%%%%%%%%%%%%%%%%%%%%%%%%%%%%%%%%%%%%%%%%

\subsection{Stochastic representation}

Suppose $Y_{k}$ has a weighted Gumbel distribution with parameters
$\mu\in\mathbb{R}$ and $\sigma>0$. That is, if  $Y\sim F_{\rm G}(\cdot;\mu,\sigma)$, by \eqref{form-mom-trunc}, $Y_k$ has CDF given by, for each $x\in\mathbb{R}$ and $k=0,1,2,\ldots$,
\begin{align}
F_{Y_k}(x)={\mathbb{E}\big( Y^k \mathds{1}_{Y\leq x}\big)\over\mathbb{E}( Y^k)}
=
\frac{\displaystyle \sum_{i=0}^{k} \binom{k}{i} \mu^{k-i}\sigma^i \, I\biggl(i;\exp\Big[-\Big(\frac{x-\mu}{\sigma}\Big)\Big],+\infty\biggr)}
{\displaystyle \sum_{i=0}^{k} \binom{k}{i} \mu^{k-i}\sigma^i \, I(i;0,+\infty)}, \label{cdf-weighted}
\end{align}
where $I(k;a,b)$ is the incomplete moments defined in \eqref{Iab}.
Note that $F_{Y_0}(x)=F_{\rm G}(x;\mu,\sigma)$.

Let $W$ be a discrete random variable, so that $W=1$ or $W=2$ or $W=3$, each with probability
\begin{align*}
p_1={2\over Z_\delta}, \quad
p_2=-{2(\mu+\sigma\gamma)\delta \over Z_\delta},
\quad
p_3={[\sigma^2{\pi^2\over 6}+(\mu+\sigma\gamma)^2]\delta^2 \over Z_\delta },
\end{align*}
respectively, where $[\delta>0 \, \text{and} \, \mu+\sigma\gamma<0]$ or $[\delta<0 \, \text{and} \, \mu+\sigma\gamma>0]$, and $Z_\delta$ is as in \eqref{partition}. It is straightforward see that $p_1+p_2+p_3=1$.

Assume that
$$T
=
\sum_{k=1}^{3}{1\over k}\, Y_{k-1} \delta_{W,k} \delta_{W,l},
\quad l=1,2,3,
$$
and that $W$ is independent of $Y_{l}$, for each $l=1,2,3$. Here $\delta_{x,y}$ is the Kronecker delta  function, i.e., this function is 1 if the variables are equal, and 0 otherwise.

\begin{proposition}
The following holds
\begin{align}\label{claim-stochastic-rep}
X=WT \ \
\text{if and only if} \ \ X\sim F_{\rm BG}(\cdot;\mu,\sigma,\delta).
%	\\
%\text{Conversely, if} \ X\sim  F_{\rm BG}(\cdot;\mu,\sigma,\delta) \ \text{then} \ X=WT.
\end{align}
\end{proposition}
\begin{proof}
By Law of total probability and by independence, we get
\begin{align*}
\mathbb{P}(X\leqslant x)
=
\mathbb{P}(WT\leqslant x)
&=
\sum_{l=1}^{3}
\mathbb{P}(WT\leqslant x\vert W=l)\mathbb{P}(W=l)
\\[0,15cm]			
&=
\sum_{l=1}^{3}
\mathbb{P}(Y_{l-1}\leqslant x\vert W=l)\mathbb{P}(W=l)
%\\[0,15cm]			
%&
=
\sum_{l=1}^{3}
F_{Y_{l-1}}(x)\, p_l.
\end{align*}
By using the CDF of the weighted Gumbel distribution, given in \eqref{cdf-weighted}, and by definitions of $p_l$'s and $Z_\delta$, the above expression is
\begin{align*}
&=
{2\over Z_\delta} \,  F_{Y_{0}}(x)
-
{2(\mu+\sigma\gamma)\delta \over Z_\delta}\, F_{Y_{1}}(x)
+
{[\sigma^2{\pi^2\over 6}+(\mu+\sigma\gamma)^2]\delta^2 \over Z_\delta }\, F_{Y_{2}}(x)
\\[0,15cm]
&=
{2\over Z_\delta}\,
F_{\rm G}(x;\mu,\sigma)
-
{2(\mu+\sigma\gamma)\delta \over Z_\delta}\,
\frac{\displaystyle \sum_{i=0}^{1} \binom{1}{i} \mu^{1-i}\sigma^i \, I\biggl(i;\exp\Big[-\Big(\frac{x-\mu}{\sigma}\Big)\Big],+\infty\biggr)}
{\displaystyle \sum_{i=0}^{1} \binom{1}{i} \mu^{1-i}\sigma^i \, I(i;0,+\infty)}
\\[0,15cm]
&\quad
+
{[\sigma^2{\pi^2\over 6}+(\mu+\sigma\gamma)^2]\delta^2 \over Z_\delta }\,
\frac{\displaystyle \sum_{i=0}^{2} \binom{2}{i} \mu^{2-i}\sigma^i \, I\biggl(i;\exp\Big[-\Big(\frac{x-\mu}{\sigma}\Big)\Big],+\infty\biggr)}
{\displaystyle \sum_{i=0}^{2} \binom{2}{i} \mu^{2-i}\sigma^i \, I(i;0,+\infty)}
\\[0,15cm]
&=
{\displaystyle
	\big[2-\delta\mu(2-\delta\mu)\big]
	\exp\biggl\{-\exp\Big[-\Big(\frac{x-\mu}{\sigma} \Big)\Big]\biggr\}
	+
	\delta^2\sigma^2
	I\biggl(2;\exp\Big[-\Big(\frac{x-\mu}{\sigma} \Big)\Big],+\infty\biggr)
	\over
	1+\delta^2\sigma^2{\pi^2\over 6}+(\delta\mu+\delta\sigma\gamma-1)^2
}
\nonumber
\\[0,2cm]
&\quad
+
{\displaystyle
	2\delta(1-\delta\mu)
	\Biggl\{
	{(x-\mu)} \,
	\exp\biggl\{-\exp\Big[-\Big(\frac{x-\mu}{\sigma} \Big)\Big]\biggr\}
	-
	\sigma
	\Gamma\biggl(0,\exp\Big[-\Big(\frac{x-\mu}{\sigma} \Big)\Big]\biggr)\Biggr\}
	\over
	1+\delta^2\sigma^2{\pi^2\over 6}+(\delta\mu+\delta\sigma\gamma-1)^2
}
\\[0,15cm]
&\stackrel{\eqref{cdfbgev2-1}}{=}
F_{\rm BG}(x; \mu,\sigma, \delta),
\end{align*}
where $\Gamma(a,b)$ is the upper incomplete gamma function.

Then the statement in \eqref{claim-stochastic-rep} follows.
\end{proof}

%%%%%%%%%%%%%%%%%%%%%%%%%%%%%%%%%%%%%%%%%%%%%%%%%%%%%%%
%%%%%%%%%%%%%%%%%%%%%%%%%%%%%%%%%%%%%%%%%%%%%%%%%%%%%%%

\subsection{Rate of a random variable with BG distribution}
\noindent

Following Klugman et al. (1998) \cite{Klugman1998},
for a continuous random variable $X$ with density function $f_X(x)$, the rate of a random variable is given by
\begin{eqnarray*}
\tau_X=-\lim_{x\to\infty} {{\rm d} \ln\big[f_X(x)\big]\over {\rm d}x}.
\end{eqnarray*}
A simple computation shows that
$
\tau_{{\rm BG}(\mu,\sigma, \delta)}
=
{1/\sigma}.
$

In what follows we present some comparisons between the rates of random variables with known distributions: Inverse-gamma, Log-normal, Generalized-Pareto, BWeibull (see Vila et al. 2020 \cite{VC20}), BGamma (see Vila et al. 2020 \cite{VFSPO20}), BG, exponential and Normal;
	\begin{align*}
	\tau_{{\rm InvGamma}(\alpha,\sigma)}
	=
	\tau_{{\rm LogNorm}(\mu,\kappa^2)}
	&=
	\tau_{{\rm GenPareto}(\alpha,\sigma, \zeta)}
		=
	\tau_{\rm BWeibull(\alpha<1,\sigma,\delta)}
	=0
	\\[0,1cm]
	&<
	\tau_{{\rm BG}(\mu,\sigma, \delta)}=
	\tau_{\rm BWeibull(\alpha=1,\sigma,\delta)}=
	\tau_{{\rm BGamma}(\alpha,1/\sigma,\delta)}=
	\tau_{{\rm exp}(1/\sigma)}=1/\sigma
	\\[0,1cm]
	&<
	\tau_{\rm BWeibull(\alpha>1,\sigma,\delta)}
	=
	\tau_{{\rm Normal}(\mu,\kappa^2)}=+\infty.
	\end{align*}
In other words, far enough out in the tail, every BG distribution looks like an exponential distribution.

%%%%%%%%%%%%%%%%%%%%%%%%
%\subsection{Modes}
%\noindent
\subsection{Bimodality}
\noindent
%%%%%%%%%%%%%%%%%%%%%% %%%%%%%%%%%%%  %%%%%%%%%%%%%%
%%%%%%%%%%%%%% Prop 2 %%%%%%%%%%%%%%%%%%%%%%%%%%%%%%%%%%%%%
%
\begin{proposition}
	%\label{prop-modes}
	A point  $x\in\mathbb{R}$ is a mode of the {\rm BG} density \eqref{pdfbgev2} if it is a root of the following non-polynomial function:
	\begin{eqnarray} \label{g-def-0}	
g(x)=
{1\over \sigma}\biggl\{\exp\Big[-\Big({x-\mu\over\sigma}\Big)\Big]-1\biggl\}
-
{2\delta (1-\delta x)\over (1-\delta x)^2+1}.
%\\[0,2cm]
%&=
%{1\over \sigma}\exp\Big[-\Big({x-\mu\over\sigma}\Big)\Big]
%-
%{\delta^2x^2-2\delta(1+\sigma\delta)x+2(1+\sigma\delta)\over \sigma\big[(1-\delta x)^2+1\big]}. \label{g-def}
	\end{eqnarray}
\end{proposition}
%%%%%%%%%%%%%%%%%%%%%%%%%%%%%%%%%%%%%%%%%%%%%%%%%%%%%%%%%%%%%%
\begin{proof}
Taking the derivative of $f_{\rm BG}(x;\mu,\sigma)$ with respect to $x$, we have
	\begin{eqnarray}\label{f-der-relation}	
	{f'_{\rm BG}(x;\mu,\sigma)}
	=
	f_{\rm BG}(x;\mu,\sigma)
g(x).
	\end{eqnarray}
Hence, the proof follows.
\end{proof}
%

%\subsection{Bimodality}
%\noindent

Let $\mathcal{C}$ be the set formed for all $(\mu,\sigma,\delta)\in\mathbb{R}\times(0,+\infty)\times \mathbb{R}$ such that the following hold:
\begin{align}
&\delta
>\max\biggl\{1,
{1\over \sigma}
\Big[\exp\Big({\mu\over\sigma}\Big)-1\Big]\biggr\},
\label{cond1}
\\[0,2cm]
&{2\delta (1+\delta)\over (1+\delta)^2+1}
<
{1\over \sigma}
\Big\{\exp\Big({1+\mu\over\sigma}\Big)-1\Big\},
\label{cond2}
\\[0,2cm]\
&{2\delta (1-2\delta)\over (1-2\delta)^2+1}
<
{1\over \sigma}
\biggl\{\exp\Big[-\Big({2-\mu\over\sigma}\Big)\Big]-1\biggl\},
\label{cond3}
\\[0,2cm]
&{2\delta (1-3\delta)\over (1-3\delta)^2+1}
>
{1\over \sigma}
\biggl\{\exp\Big[-\Big({3-\mu\over\sigma}\Big)\Big]-1\biggl\}.
\label{cond4}
\end{align}
\begin{remark1} \label{rem-C}
By considering $\mu=\sigma=1$ and $\delta>{\rm e}-1$, we have that $(\mu,\sigma,\delta)\in \mathcal{C}$.
That is, the set $\mathcal{C}$ is non-empty.
\end{remark1}

\begin{lemma}\label{lemma-existence}
If $(\mu,\sigma,\delta)\in \mathcal{C}$ then the function $g(x)$ has at least three distinct real roots.
\end{lemma}
\begin{proof}
Since $(\mu,\sigma,\delta)\in \mathcal{C}$, a simple observation in the definition \eqref{g-def-0} of $g(x)$ shows that
\begin{itemize}
\item
$\displaystyle g(-1)
=
{1\over \sigma}
\Big\{\exp\Big({1+\mu\over\sigma}\Big)-1\Big\}
-
{2\delta (1+\delta)\over (1+\delta)^2+1}>0,
$
because of condition \eqref{cond2};
\item
$\displaystyle g(0)
={1\over\sigma} \Big\{\exp\Big({\mu\over\sigma}\Big)-1\Big\}-\delta<0$,
because of condition \eqref{cond1};
\item
$\displaystyle
g(2)
=
{1\over \sigma}
\biggl\{\exp\Big[-\Big({2-\mu\over\sigma}\Big)\Big]-1\biggl\}
-
{2\delta (1-2\delta)\over (1-2\delta)^2+1}>0
$,
because of condition \eqref{cond3};
\item
$\displaystyle
g(3)
=
{1\over \sigma}
\biggl\{\exp\Big[-\Big({3-\mu\over\sigma}\Big)\Big]-1\biggl\}
-
{2\delta (1-3\delta)\over (1-3\delta)^2+1}<0
$,
because of condition \eqref{cond4}.
\end{itemize}
Since $g(x)$ is a continuous real-valued function, by Intermediate Value Theorem, there are some points $r_1, r_2, r_3$, with $-1<r_1<0$, $0<r_2<2$ and $2<r_3<3$, so that $g(r_i)=0$, $i=1,2,3$.
\end{proof}

Let $\mathcal{D}$ be the set formed for all $x\in\mathbb{R}$ such that
\begin{eqnarray*}
2\delta^2 \,
{(1-\delta x)^2-1\over (1-\delta x)^2+1}
&<&
-{1\over \sigma^2}
\exp\Big[-\Big({x-\mu\over\sigma}\Big)\Big]
\quad \text{where} \ (\mu,\sigma,\delta)\in \mathcal{C}.
\end{eqnarray*}
\begin{remark1} \label{remark-D}
The set $\mathcal{D}$ is non-empty. To see this just take $\mu=\sigma=1$ and $\delta=2$ $(>{\rm e}-1)$. By Remark \ref{rem-C}, $(\mu,\sigma,\delta)\in \mathcal{C}$. In this case we have $\mathcal{D}=(0.132178,0.937349)$.
\end{remark1}

\begin{lemma}\label{lemma-existence-1}
If $(\mu,\sigma,\delta)\in \mathcal{C}$ then the function $g(x)$ has no root outside the interval $(-1,3)$.

Futhermore, if  $-1<r_1<0$, $0<r_2<2$ and $2<r_3<3$ are the roots of $g(x)$ found in  Lemma \ref{lemma-existence}, so that $r_2\in \mathcal{D}\subset (0,2)$, then these roots are unique.
\end{lemma}
\begin{proof}
Taking the derivative of $g(x)$ with respect to $x$, we have
\begin{eqnarray*}%\label{exp-g-der}
g'(x)
=
-{1\over \sigma^2}
\exp\Big[-\Big({x-\mu\over\sigma}\Big)\Big]
-
2\delta^2 \,
{(1-\delta x)^2-1\over (1-\delta x)^2+1}.
\end{eqnarray*}
Since, for
$x\leqslant 0$ or $x\geqslant 2/\delta$,
$
2\delta^2 \,
{[(1-\delta x)^2-1] / [(1-\delta x)^2+1]} \geqslant 0,
$
 we obtain
\begin{eqnarray*}
g'(x)
\leqslant
-{1\over \sigma^2}
\exp\Big[-\Big({x-\mu\over\sigma}\Big)\Big]<0
\quad \text{for} \
x\leqslant 0 \ \text{ou} \ x\geqslant 2/\delta.
\end{eqnarray*}
In other words, outside of the interval $(-1, 3)$ the function $g(x)$ has no zeros because this one is decreasing on $(-\infty,0]\cup[2/\delta,+\infty)$ with  $2/\delta<2$ (by condition \eqref{cond1}). This implies that the roots $-1<r_1<0$ and $2<r_3<3$, of $g(x)$,  are unique in at their respective intervals. Finally, since $r_2\in \mathcal{D}$ we have that the function $g(x)$ crosses the abscissa axis at the only point $r_2$. Therefore, we conclude that the root $r_2$ is also unique in $\mathcal{D}\subset (0,2)$.
\end{proof}

\begin{remark1}
	Let $\mu=\sigma=1$ and $\delta=2$ $(>{\rm e}-1)$. By Remark \ref{remark-D}, $\mathcal{D}=(0.132178,0.937349)$.   Computationally it can be verified that $-1<r_1= -0.0896138<0, 0<r_2= 0.389792<2$ and $2<r_3 = 2.79117<3$ are the only roots of $g(x)$, with $r_2\in \mathcal{D}$.
\end{remark1}

\begin{thm1}[Bimodality]\label{Bimodality}
If $r_2\in \mathcal{D}$ then the {\rm BG} distribution \eqref{pdfbgev2} is bimodal.
\end{thm1}
\begin{proof}
By Lemma \ref{lemma-existence-1} the function $g(x)$ in \eqref{g-def-0} has exactly three roots $r_1, r_2, r_3$ so that $r_1<r_2<r_3$.
Since $f_{\rm BG}(x;\mu,\sigma)\to 0$ as $x\to \pm\infty$, it follows that the {\rm BG} distribution \eqref{pdfbgev2} increases on the intervals $(-\infty, r_1)$
and $(r_2, r_3)$, and decreases on $(r_1, r_2)$ and $(r_3, +\infty)$. That is, $r_1$ and $r_3$ are two
maximum points and $r_2$ is the unique minimum point. Therefore, the bimodality of $f_{\rm BG}(x;\mu,\sigma)$ is guaranteed.
\end{proof}

%\smallskip

\subsection{The hazard function}
\noindent

The survival and hazard functions, denoted by SF and HR  of the BG is given
by  $S_{\rm BG}(x;\mu,\sigma, \delta) = 1 -F_{\rm BG}(x;\mu,\sigma, \delta)$ and
$
	H_{\rm BG}(x;\mu,\sigma, \delta)
	=
	{
		f_{\rm BG}(x;\mu,\sigma, \delta)
		/
		[1-F_{\rm BG}(x;\mu,\sigma, \delta)]
	}.
$

Considering $r_1, r_2, r_3$ the three distinct real roots, of function $g(x)$, found in  Lemma \ref{lemma-existence}, we state the following result.
\begin{proposition}%\label{mhr-2}
If $r_2\in \mathcal{D}$ then
the HR $H_{\rm BG}(x;\mu,\sigma,\delta)$ of the {\rm BG} distribution \eqref{pdfbgev2} has the following monotonicity properties:
	\begin{itemize}
		\item[\rm 1)] It is increasing for each $x<r_1$ or $r_2<x< r_3$,
		\item[\rm 2)] It is decreasing for each $x\in\mathcal{D}$.
	\end{itemize}
\end{proposition}
\begin{proof}
As a sub-product of the proof of Theorem \ref{Bimodality}, note that the density $f_{\rm BG}(x;\mu,\sigma,\delta)$ is increasing on the intervals $x<r_1$ or $r_2<x< r_3$.
Since $1-F_{\rm BG}(x;\mu,\sigma, \delta)$ is a decreasing function, we have that the hazard function on $x<r_1$ or $r_2<x< r_3$ is the product of two increasing and nonnegative functions, then the proof of first item follows.
	
	By Glaser (1980), for to prove the second item, it is enough to prove that the function $G_X(x)$, defined as
	\begin{eqnarray*}
	G_X(x)
	=
	-{f_{\rm BG}'(x;\mu,\sigma, \delta) \over f_{\rm BG}(x;\mu,\sigma, \delta)}\stackrel{\eqref{f-der-relation}}{=}-g(x),
	\end{eqnarray*}
	is decreasing for each $x\in\mathcal{D}$.
But this is immediate since $g(x)$ is increasing on this interval.	
\end{proof}

\subsection{Moment-Generating Function}
\noindent

\begin{thm1}\label{mgf-1}
	If $X\sim F_{\rm BG}(\cdot; \mu,\sigma, \delta)$,  $t<\min\{0, -m/\sigma\}$ and $m\in\mathbb{N}\cup\{0\}$, we have
	{%\scalefont{0.8}
	\begin{align*}
	\mathbb{E}\big[X^m\exp(tX)\big]
	&=	
	(-1)^{m+2}
	\exp(t\mu)\,
	\dfrac{
		\delta^2\sigma^{m+2} \,
		\Gamma^{(m+2)}(1-\sigma t)
		+	
		\delta\sigma^{m+1}\big[2-\delta\mu(m+2)\big] \,
		\Gamma^{(m+1)}(1-\sigma t)
	}{
		1+\delta^2\sigma^2{\pi^2\over 6}+(\delta\mu+\delta\sigma\gamma-1)^2
	}
	\\[0,5cm]
	&+
	\dfrac{	\exp(t\mu)\,
		\sum_{i=0}^{m}	
		(-1)^i
		\sigma^i\mu^{m-i}
		\big[
		2\binom{m}{i}
		-
		2\delta\mu\binom{m+1}{i}
		+
		\delta^2\mu^{2} \binom{m+2}{i}
		\big]
	\Gamma^{(i)}(1-\sigma t)
}{
		1+\delta^2\sigma^2{\pi^2\over 6}+(\delta\mu+\delta\sigma\gamma-1)^2
	},
\end{align*}
}
where $\Gamma^{(i)}(x)={{\rm d}^i \Gamma(x)/ {\rm d}x^i}$.
\end{thm1}
\begin{proof}	
	
	Taking $g(X)=X^m\exp(tX)$ in Remark \ref{eq-main},	
	\begin{align}\label{eq-fgm}
\mathbb{E}\big[X^m\exp(t\,X)\big]
=
{1 \over Z_{\delta}}\,
\big\{
2 \, \mathbb{E}\big[Y^m\exp(t\,Y)\big]
-
2\delta \, \mathbb{E}\big[Y^{m+1} \exp(t\,Y)\big]
+
\delta^2 \, \mathbb{E}\big[Y^{m+2} \exp(t\,Y)\big]
\big\},
\end{align}	
	where $ Y\sim F_{\rm G}(\cdot; 0,  \mu, \sigma)$ and    	 $Z_{\delta}$ are as in \eqref{cdfgev1} and \eqref{partition}, respectively.
	
Taking the change of variables
$
z=\exp[-({y-\mu})/{\sigma}], \
{\rm d} z= -({z/\sigma}) {\rm d}y,
\ y=\mu-\sigma\ln(z),
$
and using a binomial expansion we get, for $k=m,m+1,m+2$,
\begin{eqnarray}\label{form-mom-trunc-1}
\mathbb{E}\big[ Y^k \exp(tY)\big]
&=&
\int_{0}^{+\infty}
[\mu-\sigma\ln(z)]^k z^{(1-\sigma t)-1} \exp(-z)\, {\rm d}z, \quad k\in\mathbb{N}\cup\{0\}
\nonumber
\\[0,2cm]
&=&
\exp(\mu t)
\sum_{i=0}^{k} \binom{k}{i} (-1)^i \mu^{k-i}\sigma^i \int_{0}^{+\infty}
\ln^i(z)\, z^{(1-\sigma t)-1} \exp(-z)\, {\rm d}z.
\end{eqnarray}

\smallskip
We claim that
\begin{eqnarray}\label{claim-1}
\int_{0}^{+\infty}
\ln^i(z)\, z^{(1-\sigma t)-1} \exp(-z)\, {\rm d}z
=
\Gamma^{(i)}(1-\sigma t), \quad i=0,1,\ldots,k; \ k=m,m+1,m+2.
\end{eqnarray}
Note that, by combining \eqref{claim-1}, \eqref{form-mom-trunc-1} and \eqref{eq-fgm}, the proof of theorem follows.

Indeed, since ${\partial^j z^{\alpha-1}\over\partial \alpha^j}=\ln^j (z)\, z^{\alpha-1}$,  $j=1,2,\ldots i$,
\begin{eqnarray}\label{id-integral}
	\int_{0}^{+\infty}
	\ln^i(z)\, z^{(1-\sigma t)-1} \exp(-z)\, {\rm d}z
	=
	\int_{0}^{+\infty} 	\dfrac{\partial^j}{\partial(1-\sigma t)^j} \big[\ln^{i-j}(z)\,  z^{(1-\sigma t)-1} \exp(-z)\big]\, {\rm d}z.
\end{eqnarray}

Let us note that the following conditions hold true.
\begin{itemize}
	\item The derivatives
	\begin{eqnarray*}
		\dfrac{\partial^j}{\partial(1-\sigma t)^j}
		\big[\ln^{i-j}(z)\,  z^{(1-\sigma t)-1} \exp(-z)\big]
		=
		\ln^i(z)\, z^{(1-\sigma t)-1} \exp(-z)
	\end{eqnarray*}
	exists and are continuous in $1-\sigma t$ for all $z$ and all $1-\sigma t$, for $j=1,2,\ldots i$.
	\item
	By using the following inequalities
	\begin{eqnarray*}
	&&\vert \ln(z)\vert
	\leqslant {1\over z}\,   \mathds{1}_{\{x>0:x<1\}}(z)+  z\, \mathds{1}_{\{x>0:x\geqslant 1\}}(z),
	\\[0,2cm]
	&&	\vert \ln^i(z)\vert
	\leqslant {1\over z^{i-1}}\,
	\mathds{1}_{A_i}(z)+  z^{i-1}\, \mathds{1}_{B_i}(z), \quad i=2,3,\ldots,k,
	\end{eqnarray*}
where $A_i=\big\{x>0:x<{\rm e}^{i W({i-1\over i})/(i-1)}\big\}$, $B_i=\big\{x>0: x\geqslant {\rm e}^{i W({i-1\over i})/(i-1)}\big\}$, and $W$ is the product logarithm function (or Lambert $W$ function), we get
	\begin{align*}
		&
		\biggl\vert
		\dfrac{\partial^j}{\partial(1-\sigma t)^j}
		\big[\ln^{i-j}(z)\,  z^{(1-\sigma t)-1} \exp(-z)\big]
		\biggr\vert
		=
		\big\vert\ln^i(z)\big\vert\, z^{(1-\sigma t)-1} \exp(-z)
		\\[0,3cm]
		&\leqslant
		\begin{cases}
		z^{(1-\sigma t)-1} \exp(-z) & \text{if} \ i=0,
		\\[0,2cm]
		\big[
		z^{(1-\sigma t)-2}\,   \mathds{1}_{\{x>0:x<1\}}(z)+z^{(1-\sigma t)} \, \mathds{1}_{\{x>0:x\geqslant 1\}}(z) \big]\exp(-z) & \text{if} \ i=1,
		\\[0,2cm]
		\big[		
		z^{(1-\sigma t)-i}\,
		\mathds{1}_{A_i}(z)
		+
		z^{(1-\sigma t)+i-2}\,
		\mathds{1}_{B_i}(z)
		\big] \exp(-z)
		& \text{if} \ i=2,3,\ldots,k,
		\end{cases}
		\eqqcolon
		G(z,1-\sigma t).		
	\end{align*}
	Furthermore, the integral
	\begin{align*}
		\int_{0}^{+\infty}  G(z,1-\sigma t) \, {\rm d}z
		\leqslant
\begin{cases}
\Gamma(1-\sigma t) & \text{if} \ i=0, \ t<1/\sigma,
\\[0,2cm]
\Gamma(-\sigma t)
+\Gamma(2-\sigma t)
& \text{if} \ i=1, \ t<0,
\\[0,2cm]		
\Gamma(2-\sigma t-i)
+
\Gamma(i-\sigma t)
& \text{if} \ i=2,3,\ldots,k, \ t<(2-k)/\sigma,
\end{cases}
\end{align*}
is finite, for $k=m,m+1,m+2$.
	\item The integral
	$
	\int_{0}^{+\infty} \ln^{i-j}(z)\,  z^{(1-\sigma t)-1} \exp(-z)\, {\rm d}z
	$
	exists because of last item above.
\end{itemize}
\smallskip
Under the above three conditions, by Leibniz integral rule,  we can interchange the derivative with the integral in \eqref{id-integral}. Hence, for  $j=1,2,\ldots i$,
\begin{eqnarray*}
\int_{0}^{+\infty}
\ln^i(z)\, z^{(1-\sigma t)-1} \exp(-z)\, {\rm d}z
	=
	\dfrac{{\rm d}^j}{{\rm d}(1-\sigma t)^j}
\int_{0}^{+\infty} \ln^{i-j}(z)\,  z^{(1-\sigma t)-1} \exp(-z)\, {\rm d}z.
\end{eqnarray*}
Letting $j=i$, we obtain
\begin{eqnarray*}
\int_{0}^{+\infty}
\ln^i(z)\, z^{(1-\sigma t)-1} \exp(-z)\, {\rm d}z
&=&
%	\dfrac{{\rm d}}{{\rm d}(1-\sigma t)}
%\int_{0}^{+\infty} \ln^{i-1}(z)\,  z^{(1-\sigma t)-1} \exp(-z)\, {\rm d}z \nonumber
%\\[0,15cm]
%&=&
%	\dfrac{{\rm d}^2}{{\rm d}(1-\sigma t)^2}
%\int_{0}^{+\infty} \ln^{i-2}(z)\,  z^{(1-\sigma t)-1} \exp(-z)\, {\rm d}z \nonumber
%\\[0,15cm]
%&\vdots& \nonumber
%\\[0,15cm]
%&=&
	\dfrac{{\rm d}^i}{{\rm d}(1-\sigma t)^i}
\int_{0}^{+\infty}  z^{(1-\sigma t)-1} \exp(-z)\, {\rm d}z \nonumber
\\[0,15cm]
&=&
\Gamma^{(i)}(1-\sigma t), \quad i=0,1,\ldots,
\end{eqnarray*}
where in the last line we used the definition of gamma function, $\Gamma(\alpha)=\int_{0}^{+\infty}  z^{\alpha-1} \exp(-z)\, {\rm d}z$. Then the claimed \eqref{claim-1} follows.

Thus, this complete the proof of theorem.
\end{proof}

By taking $m=0$ in Theorem \ref{mgf-1} we obtain a clossed expression for the moment-generating function of BGumbel distribution.
\begin{corollary}\label{mgf}
	If $X\sim F_{\rm BG}(\cdot; \mu,\sigma, \delta)$  then, the moment-generating function $M_X(t)=\mathbb{E}[\exp(tX)]$ with $t<0$, is given by
	\begin{eqnarray*}
		M_X(t)
		=
		{
			\exp(\mu t) \Gamma(1-\sigma t)
			\big[
			2-2\mu\delta+\mu^2\delta^2+2\sigma\delta(1-\mu\delta) \psi(1-\sigma t)+
			{\sigma^2\delta^2\over \Gamma(1-\sigma t)} \,
			\Gamma^{(2)}(1-\sigma t)
			\big]
			\over
			1+\delta^2\sigma^2{\pi^2\over 6}+(\delta\mu+\delta\sigma\gamma-1)^2	
		},
	\end{eqnarray*}
	where $\psi(x)=\Gamma'(x)/ \Gamma(x)$ is the digamma function and $\Gamma^{(2)}(1-\sigma t)={{\rm d}^2\Gamma(x)/ {\rm d} x^2}$.
\end{corollary}

\begin{remark1}
	Letting $\delta=0$ in Corollary \ref{mgf} we obtain the known formula $M_Y(t)=\exp(\mu t) \Gamma(1-\sigma t)$, where $Y\stackrel{d}{=}X\sim F_{\rm G}(\cdot; 0,  \mu, \sigma)$. Note that, by a  by-product of proof of Theorem \ref{mgf-1}, in this case it is sufficient take $t<1/\sigma$ instead $t<0$.
\end{remark1}
\begin{remark1}
The characteristic function of $X\sim F_{\rm BG}(\cdot; \mu,\sigma, \delta)$, denoted by $\phi_X(t)$, can be obtained from the moment-generating function by the relation $M_X(t) =\phi_X(-it)$.
\end{remark1}

\subsection{Moments}
%%%%%%%%%%%%% Prop 5 %%%%%%%%%%%%%%%%%%%%%%%%%%%%%%%%%%%%%%%%%%%%%%
\noindent

\begin{thm1}\label{car-moments}
	If $X\sim F_{\rm BG}(\cdot; \mu,\sigma, \delta)$  then
	\begin{eqnarray*}
	\mathbb{E}(X^k)
	=
	{\displaystyle
		\delta^2\sigma^{k+2}I(k+2;0,+\infty)
		-	
		\delta\sigma^{k+1}\big[2-\delta\mu(k+2)\big]I(k+1;0,+\infty)
		\over
		1+\delta^2\sigma^2{\pi^2\over 6}+(\delta\mu+\delta\sigma\gamma-1)^2
	}
	\\[0,2cm]
	+
	{%\displaystyle
		\sum_{i=0}^{k}	
		\sigma^i\mu^{k-i}
		\big[
		2\binom{k}{i}
		-
		2\delta\mu\binom{k+1}{i}
		+
		\delta^2\mu^{2} \binom{k+2}{i}
		\big] I(i;0,+\infty)
		\over
		1+\delta^2\sigma^2{\pi^2\over 6}+(\delta\mu+\delta\sigma\gamma-1)^2
	},
	\end{eqnarray*}
	where $I(i;a,b)$ is defined in  \eqref{Iab}.
\end{thm1}
\begin{proof}
	From Remark \ref{eq-main} with $g(X)=X^k$, we obtain	
	\begin{eqnarray}\label{eq-1}
	\mathbb{E}(X^k)
	&=&
	{1 \over Z_{\delta}}\,
	\big[
	2 \, \mathbb{E}(Y^{k})
	-
	2\delta \, \mathbb{E}(Y^{k+1})
	+
	\delta^2 \, \mathbb{E}(Y^{k+2})
	\big],
	\end{eqnarray}
	where $ Y\sim F_{\rm G}(\cdot; 0,  \mu, \sigma)$ is as in \eqref{cdfgev1}.
	
	By taking $x\longrightarrow +\infty$ in \eqref{form-mom-trunc}, Lebesgue dominated convergence theorem gives
	\begin{eqnarray}\label{eq-2}
	\mathbb{E}(Y^k)
%	&=&
	=
	\lim_{x\to +\infty}
	\mathbb{E}\big( Y^k \mathds{1}_{Y\leq x}\big) %\nonumber
%	\\[0,1cm]
%	&=&
	=
	\sum_{i=0}^{k} \binom{k}{i} \mu^{k-i}\sigma^i \, I(i;0,+\infty).
	\end{eqnarray}
	
	By combining \eqref{eq-1} and \eqref{eq-2}, the proof follows.
\end{proof}

%%%%%%%%%%%%%%%%%%%%%%%%%%%%%%%%%%%%%%%  Prop 8  %%%%%%%%%%%%%%%%%%%
%%%%%%%%%%%%%%%%%%%
\begin{corollary}%\label{corollary-moments}
	If  $X\sim F_{\rm BG}(\cdot; \mu,\sigma,\delta)$  then
	\begin{align}\label{mean}
		\mathbb{E}(X)
	&=
	{
\delta^2\sigma^3\big[2\zeta(3) + \gamma^3 + {\gamma \pi^2\over 2}\big]
-
\delta\sigma^2(2-3\delta\mu)\big(\gamma^2+{\pi^2\over 6}\big)
+
\mu\big[2-\delta\mu(2-\delta\mu)\big]
+
\sigma\big[2-\delta\mu(4-3\delta\mu)\big]\gamma
		\over
		1+\delta^2\sigma^2{\pi^2\over 6}
		+
		(\delta\mu+\delta\sigma\gamma-1)^2
	},
\\[0,2cm]
\label{2moment}
	\mathbb{E}(X^2)
	&=
	{
     \delta^2\sigma^4\big[8\gamma \zeta(3) + \gamma^4 + \gamma^2 \pi^2 + {3 \pi^4\over 20}\big]
     -
     2\delta\sigma^3(1-2\delta\mu)\big[2\zeta(3) + \gamma^3 + {\gamma \pi^2\over 2}\big]
     +
     \mu^2\big[2-\delta\mu(2-\delta\mu)]
		\over
		1+\delta^2\sigma^2{\pi^2\over 6}+(\delta\mu+\delta\sigma\gamma-1)^2
	}
	\nonumber\\
	&+
	{	
     2\sigma\mu\big[2-\delta\mu(3-2\delta\mu)\big]\gamma
+
2\sigma^2\big[1-3\delta\mu(1-\delta\mu)\big] \big(\gamma^2+{\pi^2\over 6}\big)
		\over
		1+\delta^2\sigma^2{\pi^2\over 6}+(\delta\mu+\delta\sigma\gamma-1)^2
	},
\\[0,2cm]
	\mathbb{E}(X^3)
	&=
{
\delta^2\sigma^5\big[20\gamma^2\zeta(3) + {10 \pi^2 \zeta(3)\over 3} + 24\zeta(5) + \gamma^5 + {5\gamma^3 \pi^2\over 3} + {3 \gamma \pi^4\over 4}\big]	 
\over
	1+\delta^2\sigma^2{\pi^2\over 6}+(\delta\mu+\delta\sigma\gamma-1)^2
}
\nonumber	\\
&+
{	
\mu^3 \big[2-\delta\mu(2-\delta\mu)\big]
+
\sigma\mu^2\big[6-\delta\mu(8-5\delta\mu)\big]\gamma
+
2\sigma^2\mu\big[3-\delta\mu(6-5\delta\mu)\big] (\gamma^2+{\pi^2\over 6})
	\over
	1+\delta^2\sigma^2{\pi^2\over 6}+(\delta\mu+\delta\sigma\gamma-1)^2
}
\nonumber	\\
&+
{	
2\sigma^3\big[1-\delta\mu(4-5\delta\mu)\big]\big[2\zeta(3) + \gamma^3 + {\gamma \pi^2\over 2}\big]
-
\delta\sigma^4(2-5\delta\mu)\big[8\gamma \zeta(3) + \gamma^4 + \gamma^2 \pi^2 + {3 \pi^4\over 20}\big]
	\over
	1+\delta^2\sigma^2{\pi^2\over 6}+(\delta\mu+\delta\sigma\gamma-1)^2
},	\nonumber	
	\end{align}
	where
	$\gamma$  is the Euler-Mascheroni constant and $\zeta(s)$ is the
	Riemann zeta function.
\end{corollary}
%%%%%%%%%%%%%%%%%%%%%%%%%%%%%%%%%
\begin{proof} By combining Theorem \ref{car-moments} with the following values for the
improper integral $I(k;0,+\infty)$ in \eqref{Iab}, with $k =0,1,2,3,4$ and $5$,
	\begin{eqnarray}\label{improper-int}
		\begin{array}{lllll}
		I(0;0,+\infty)&=&1,
		\\[0,1cm]
		I(1;0,+\infty)&=&\gamma,
		\\[0,1cm]
		I(2;0,+\infty)&=&\gamma^2+{\pi^2\over 6},
		\\[0,1cm]
		I(3;0,+\infty)&=&2\zeta(3) + \gamma^3 + {\gamma \pi^2\over 2},
		\\[0,1cm]
		I(4;0,+\infty)&=&8\gamma \zeta(3) + \gamma^4 + \gamma^2 \pi^2 + {3 \pi^4\over 20},
		\\[0,1cm]
		I(5;0,+\infty)&=& 20\gamma^2\zeta(3) + {10 \pi^2 \zeta(3)\over 3} + 24\zeta(5) + \gamma^5 + {5\gamma^3 \pi^2\over 3} + {3 \gamma \pi^4\over 4},
		\end{array}
	\end{eqnarray}
the proof follows.	
\end{proof}

For the variance of the BGumbel distribution
\begin{equation}\label{variance}
    {\rm Var}(X)=\mathbb{E}(X^2)-\big[\mathbb{E}(X)\big]^2,
\end{equation}
just replace (\ref{mean}) and (\ref{2moment}) in (\ref{variance}).

\begin{remark1}[Standardized moments]
Let $X\sim F_{\rm BG}(\cdot; \mu,\sigma,\delta)$ with $\mathbb{E}(X)$ and $\sqrt{{\rm Var}(X)}>0$, both finite.
Newton's binomial expansion gives
\begin{eqnarray}\label{exp-bin}
	\mathbb{E}\biggl[{X-\mathbb{E}(X)\over \sqrt{{\rm Var}(X)}}\, \biggr]^n
	=
	{(-1)^{n}\big[\mathbb{E}(X)\big]^{n} \over \big[{\rm Var}(X) \big]^{n/2} }\, \sum_{k=0}^{n}\binom{n}{k} (-1)^{-k}\, \big[\mathbb{E}(X)\big]^{-k}\, \mathbb{E}(X^k).
\end{eqnarray}
By combining \eqref{exp-bin}, Theorem \ref{car-moments}, identities in \eqref{improper-int} and the following relation
\begin{eqnarray*}
\begin{array}{lllll}
I(6;0,+\infty) = 20 \gamma (2 \gamma ^2 + \pi^2) \zeta(3) + 40 [\zeta(3)]^2 + 144 \gamma \zeta(5) + \gamma ^6 + {5 \gamma ^4 \pi^2\over 2} + {9 \gamma ^2 \pi^4\over 4} + {61 \pi^6\over 168},
\end{array}
\end{eqnarray*}
closed formulas for
skewness and  kurtosis of the random variable $X\sim F_{\rm BG}(\cdot; \mu,\sigma,\delta)$  it is possible to obtain.
\end{remark1}

\section{Graphical illustrations and simulation results} \label{Sec:4}
\noindent

\subsection{Graphical illustration}
\noindent

Here, the flexibility of the BG model is shown. Note that the  BG model can be unimodal or bimodal.
Figures \ref{F1}  and \ref{F2} show how the BG density function is influenced by the shape parameters $\delta$ and $\mu$, since $ \sigma $ is a scale parameter.
Figure \ref{F1} shows that the bimodality of the BG model is more evident the greater the absolute value of $ \delta $. The density tends to become unimodal when the values of $ \delta $ and $ \mu $ are far apart. In the bimodal case, the left mode is smaller and the opposite for the right mode as $ \mu $ increases, see Figure \ref{F2}. All graphics were generated with our bgumbel package, \cite{Brom}.
\begin{figure}[!htbp]
	\centering
	\includegraphics[width=0.73\linewidth]{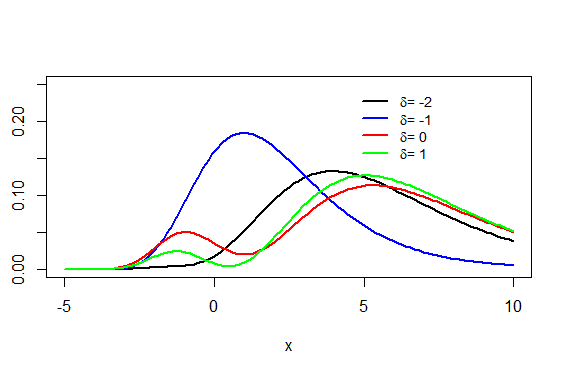}
	\vspace{-0.8cm}
	\caption{BG density graphs for
		$\mu=1$, $\sigma=2$ and $\delta$ varying
		as shown in the caption.
	}
	\label{F1}
\end{figure}

\begin{figure}[!htbp]
	\centering
	\includegraphics[width=0.73\linewidth]{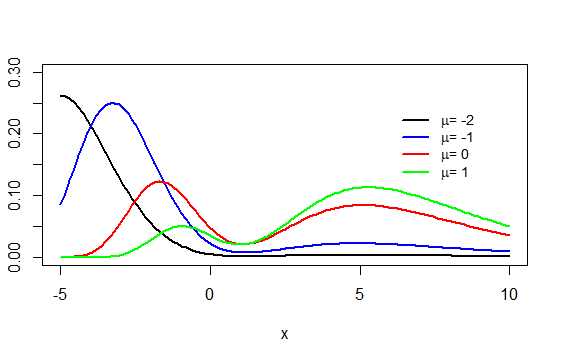}
	\vspace{-0.8cm}
	\caption{BG density graphs for
		$\delta=1$, $\sigma=2$ and $\mu$ varying
		as shown in the caption.}
	\label{F2}
\end{figure}

%\begin{figure}[!htbp]
%	\centering
%	\includegraphics[width=0.8\linewidth]{mu0d1sigma.png}
%		\vspace{-0.8cm}
%	\caption{ BG density graphs  for
%	 $\mu=0$, $\delta=1$ and $\sigma$ varying
%	 as shown in the caption.}
%	\label{F3}
%\end{figure}

\subsection{Simulation results}
\noindent

In this section, we compare the sample mean and  variance  with  the  population mean and variance  (\ref{mean}) and (\ref{variance}). In addition,   we compared the cumulative distribution function given in (\ref{cdfbgev2-1}) with the empirical distribution. For that, we present a Markov chain Monte Carlo (MCMC) simulation study with various sample sizes.
In the MCMC simulation method we considered the Metropolis-Hastings  method  as pseudo-random numbers (PRN) generator \cite{metropolis},\cite{hastings}.
We chose this method instead of the inverse transform method, because the  BG cumulative distribution (\ref{CDF}) does not admit analytical inversion.
To generate  PRN  was used the MCMCpack package \cite{MCMCpack}  available in the R program \cite{R}.
The mean, variance  and the corresponding bias of the sample estimates  were computed over $10^5$ iterations of MCMC. These results are shown in
Tables \ref{comparativemoments}.
%The population mean and variance were calculated using the expressions (\ref{mean}), (\ref{2moment}) and (\ref{variance}).
The performance of the MCMC algorithm was tested by varying the number of iterations, n.
Graphically, the convergence of the Markov Chain is shown in Figure  \ref{setparameters1}, for the first set of parameters in Table \ref{comparativemoments}.
For the other parameters sets , in the Table \ref{comparativemoments},
the performance of the algorithm is illustrated in  Appendix A, Figures \ref{figa1} - \ref{figa9}.  The value of n varies from 1000 to 100000 for each parameters set.
In all these figures the transition probability density function (TPDF) ,defined by the MCMCpack package, is in blue and the estacionary density in red.
We can conclude that the algorithm, \cite{Brom}, is efficient for n greater than
100000.
%%%%%%%%%%%%%%%%%%%%%%%%%%%%%%%%%%%%%%%%%%%%%%%%%%%%%
\begin{table}[!htbp]
	\centering
	\caption{ Comparison between the sample moments and population moments ($n = 10^5$).}
	%\centering
	\begin{tabular}{cccrrrccc}
		\hline
		$\mu$ & $\sigma$ & $\delta$ & sample mean &   $E(X)$ & Bias(mean) &  variance sample & $Var(X)$ & Bias (variance)   \\
		\hline
		$-$2 & 1 & $-$1 & $-$1.0931 & $-$1.0640 & $-$0.0291 & 4.9268 & 5.0126 & $-$0.0858
		\\
		$-$1 & 2 & $-$1 & 3.9504 & 4.0170 & $-$0.0666 & 17.126 & 18.016 & $-$0.8902
		\\
		$-$1 & 2 & $-$2 & 3.9901 & 3.9909 & $-$0.0008 & 21.032 & 21.575 & $-$0.5435
		\\
		$-$2 & 2 & $-$1 & 2.0065 & 1.9512 & 0.0553 & 24.104 & 24.592 & $-$0.4883
		\\
		\hline
	\end{tabular}
	\label{comparativemoments}
\end{table}

\begin{figure}[!htbp]
	\centering
	\includegraphics[width=0.73\linewidth]{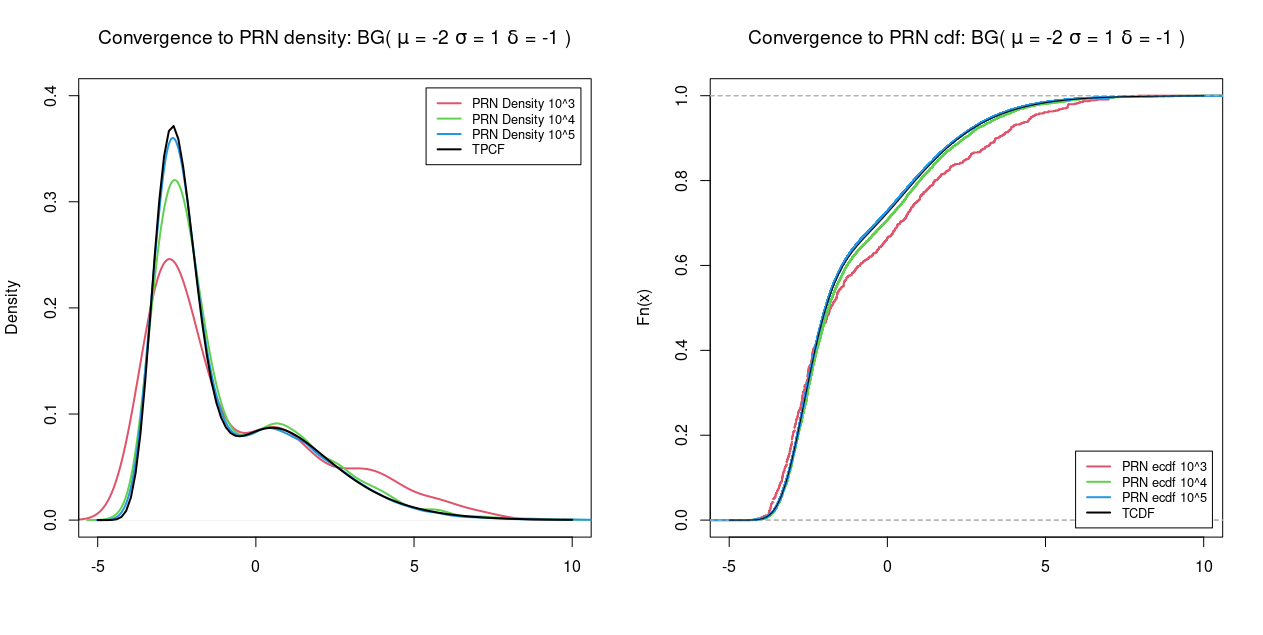}
	\vspace{-1cm}
	\caption{  Transition probability density function (TPDF)(left) and  transition cumulative distribution function (TCDF)(right), varying $n$ from $10^3$ to $10^5$.}
	\label{setparameters1}
\end{figure}

%\smallskip
%%%%%%%%%%%%%%%%%%%%%%%%%%%%%%%%%%%%%%%%%%%%%%%%%%%%%%%%%%%%%%%%%%%%%%%%%%%%%%%
\section{Maximum likelihood estimation}
\label{Sec:5}
\noindent

In this section, we determine the maximum likelihood estimates (MLEs) of the parameters of the BG distribution.
Let $X\sim F_{\rm BG}(\cdot; \mu,\sigma,\delta)$ be a random variable with PDF $f _{\rm BG} (x; \Theta) $ defined in \eqref{pdfbgev2}, where $\Theta=(\mu,\sigma,\delta)$. Let $ X_1, \dots, X_n $ be a random sample of $X$ and
$x= (x_1, \dots, x_n) $ the corresponding observed values of the random sample $ X_1, \dots, X_n $.
The log-likelihood function  for $\Theta=(\mu, \sigma, \delta)$ is given by
\begin{align}\label{fv2}
\! \ell(\Theta;x)
=
-n\ln Z_\delta
-
n\ln \sigma
+
\sum_{i=1}^{n}\left\{
\ln\big[(1-\delta x_i)^2+1\big]
-
\Big(\frac{x_i-\mu}{\sigma}\Big)- \exp\Big[-\Big(\frac{x_i-\mu}{\sigma} \Big)\Big] \right\}.
\end{align}

The first-order partial derivatives of $ Z_{\delta} $ are given by
\begin{align*}%\label{parcialz}
& \textstyle
\frac{\partial Z_{\delta}}{\partial \mu}
= \textstyle
2\delta\big[\delta(\mu+\sigma\gamma)-1\big];
\\
& \textstyle
\frac{\partial Z_{\delta}}{\partial \sigma}
=
\frac{\sigma\delta \pi^2}{3}
+
2\delta\gamma\big[\delta(\mu+\sigma\gamma)-1\big];
\\
& \textstyle
\frac{\partial Z_{\delta}}{\partial \delta}
=
\frac{\sigma\delta \pi^2}{3}
+
2\big[\delta(\mu+\sigma\gamma)-1\big](\mu+\sigma\gamma).
\end{align*}
Then, the mle of $\mu, \sigma, \delta$ are the solutions of the following system of equations
\begin{align*}%\label{1parcialL2}
\frac{\partial \ell(\Theta;x)}{\partial \mu}&= -\frac{n}{Z_{\delta}}\frac{\partial Z_{\delta}}{\partial \mu}+  \frac{n}{\sigma} -\frac{1}{\sigma}\sum_{i=1}^{n}
\exp\Big[-\Big(\frac{x_i-\mu}{\sigma} \Big)\Big]=0;
\\
\frac{\partial \ell(\Theta;x)}{\partial \sigma}
&=
-\frac{n}{Z_{\delta}}\frac{\partial Z_{\delta}}{\partial \sigma}
-
\frac{n}{\sigma}
-
\frac{1}{\sigma^2}\sum_{i=1}^{n}
(x_i - \mu)
\left\{1-\exp\Big[-\Big(\frac{x_i-\mu}{\sigma} \Big)\Big] \right\}=0;
\\
\frac{\partial \ell(\Theta;x)}{\partial \delta}
&=
-\frac{n}{Z_{\delta}}
\frac{\partial Z_{\delta}}{\partial \delta}
-2
\sum_{i=1}^{n}
{x_i
	(1-\delta x_i)\, \over
	(1-\delta x_i)^2 +1}
=0.
\end{align*}

The MLEs  $\widehat\mu, \widehat\sigma$ and $\widehat\delta$ of $\mu, \sigma$ and $\delta$
are defined as the values of $\widehat\mu, \widehat\sigma$ and $\widehat\delta$ that maximize the log-likelihood
function in (\ref{fv2}). There will be, in general, no closed form for the MLE
and their obtention will need, in practice, numerical methods.

Since the second order partial derivatives of $ Z_{\delta} $ are
\begin{align*}%\label{parcialz2}
\begin{array}{lllll}
&
\frac{\partial^2 Z_{\delta}}{\partial \mu^2}
= 2\delta^2;
&
\frac{\partial^2 Z_{\delta}}{\partial \sigma^2}
= \frac{\delta \pi^2}{3}+2\delta^2\gamma^2;
\\[0,2cm]	
&
\frac{\partial^2 Z_{\delta}}{\partial \delta^2}
= \frac{\sigma \pi^2}{3}+2(\mu+\sigma\gamma)^2;
&
\frac{\partial^2 Z_{\delta}}{\partial \mu \partial\sigma}
=
\frac{\partial^2 Z_{\delta}}{\partial\sigma \partial \mu}
=
2\delta^2\gamma;
\\[0,2cm]	
&
\frac{\partial^2 Z_{\delta}}{\partial \mu \partial\delta}
=
\frac{\partial^2 Z_{\delta}}{\partial\delta \partial \mu}
=
4\delta(\mu+\sigma\gamma);
&
\frac{\partial^2 Z_{\delta}}{\partial \sigma \partial \delta}
=
\frac{\partial^2 Z_{\delta}}{\partial \delta \partial \sigma}	
=
\frac{\sigma \pi^2}{3}+ 4\delta\gamma(\mu+\sigma\gamma);
\end{array}	
\end{align*}
the second-order partial derivatives of the log-likelihood function  $\ell(\Theta;x)$ are given by
\begin{align*}%\label{2parcialL}
\frac{\partial^2 \ell(\Theta;x)}{\partial \mu^2}
&=
-D_{\mu,\mu}(\Theta;n)
-
\frac{1}{\sigma^2}
\sum_{i=1}^{n} \exp\Big[-\Big(\frac{x_i-\mu}{\sigma} \Big)\Big];
\\
\frac{\partial^2 \ell(\Theta;x)}{\partial \sigma^2}
&=
-D_{\sigma,\sigma}(\Theta;n)
+
\frac{n}{\sigma^2}
+
\frac{1}{\sigma^3}
\sum_{i=1}^{n}
(x_i - \mu)
\left\{
2- \Big[2-\Big(\frac{x_i-\mu}{\sigma}\Big) \Big]  \exp\Big[-\Big(\frac{x_i-\mu}{\sigma} \Big)\Big]
\right\}
;
\\
\frac{\partial^2 \ell(\Theta;x)}{\partial \delta^2}
&=
-D_{\delta,\delta}(\Theta;n)
-2
\sum_{i=1}^{n}	
{x_i^2 \big[(1-\delta x_i)^2 -1\big] \over
	\big[(1-\delta x_i)^2 +1\big]^2}.
\end{align*}
Already, the second-order mixed derivatives of $\ell(\Theta;x)$ can be written as
\begin{align*}%\label{2parcialL1}
\frac{\partial^2 \ell(\Theta;x)}{\partial \mu \partial \sigma}
&=
\frac{\partial^2 \ell(\Theta;x)}{\partial \sigma \partial \mu}
=
-D_{\mu,\sigma}(\Theta;n)
-
\frac{n}{\sigma^2}
+
\frac{1}{\sigma^2}
\sum_{i=1}^{n} \Big[1-\Big(\frac{x_i-\mu}{\sigma}\Big) \Big]  \exp\Big[-\Big(\frac{x_i-\mu}{\sigma} \Big)\Big];
\\
\frac{\partial^2 \ell(\Theta;x)}{\partial \mu \partial \delta}
&=
\frac{\partial^2 \ell(\Theta;x)}{\partial \delta \partial \mu }
=
-D_{\mu,\delta}(\Theta;n);
\\
\frac{\partial^2 \ell(\Theta;x)}{\partial\sigma \partial \delta}
&=
\frac{\partial^2 \ell(\Theta;x)}{ \partial \delta \partial\sigma}
=
-D_{\sigma,\delta}(\Theta;n),
\end{align*}
where
\begin{align}\label{D-def}
\textstyle
D_{u,v}(\Theta;n)=\frac{\partial}{\partial u} \big(\frac{n}{Z_{\delta}}\frac{\partial Z_{\delta}}{\partial v} \big)
=
\frac{n}{Z_{\delta}}
\big(
\frac{\partial^2 Z_{\delta}}{\partial u\partial v}
-
{1\over Z_{\delta}}\frac{\partial Z_{\delta}}{\partial u} \frac{\partial Z_{\delta}}{\partial v}
\big), \quad u,v\in\{\mu,\sigma,  \delta\}.
\end{align}

Thus, the elements of the Fisher information matrix $I_X(\Theta)$ are defined by
\begin{eqnarray*}
	[I_X(\Theta)]_{jk}=\mathbb{E}\biggl[{\partial \ln f_{\rm BG}({X};{\Theta})\over \partial\theta_j}\ {\partial \ln f_{\rm BG}({X};{\Theta})\over \partial\theta_k}\biggr],
\end{eqnarray*}
for $\theta_j,\theta_k\in\{\mu,\sigma,  \delta\}$ and $j,k=1,2,3$.  Under known regularity conditions,   the elements $[I_X(\Theta)]_{jk}$ are written as
$\mathbb{E}\big[-{\partial^2 \ln f_{\rm BG}({X};{\Theta})\over \partial\theta_j\partial\theta_k}\big]$.
Hence, by considering the following notations:
\begin{eqnarray*}
	\begin{array}{lllll}
		F_1(X)&=& \exp\big[{-\big(\frac{X-\mu}{\sigma}\big)}\big],
		\\[0,3cm]
		F_2(X)&=&
		\big[1-\big(\frac{X-\mu}{\sigma}\big)\big]
		\exp\big[{-\big(\frac{X-\mu}{\sigma}\big)}\big],
		\\[0,3cm]
		F_3(X)&=&
		(X - \mu)
		\big\{
		2-\big[2-\big(\frac{X-\mu}{\sigma}\big) \big]
		\exp\big[{-\big(\frac{X-\mu}{\sigma}\big)}\big]
		\big\},
		\\[0,3cm]
		F_4(X)&=&
		{X^2 \big[(1-\delta X)^2 -1\big] /
			\big[(1-\delta X)^2 +1\big]^2},
	\end{array}
\end{eqnarray*}
we get
\begin{align*}
I_X(\Theta)=
\scalemath{0.95}{
	\begin{bmatrix}
	D_{\mu,\mu}(\Theta;1) + \frac{1}{\sigma^2} \mathbb{E}[F_1(X)]
	&
	D_{\mu,\sigma}(\Theta;1)
	+
	\frac{1}{\sigma^2}
	-
	\frac{1}{\sigma^2} \mathbb{E}[F_2(X)]
	&
	D_{\mu,\delta}(\Theta;1)
	\\[0,3cm]
	D_{\mu,\sigma}(\Theta;1)
	+
	\frac{1}{\sigma^2}
	-
	\frac{1}{\sigma^2} \mathbb{E}[F_2(X)]
	&
	D_{\sigma,\sigma}(\Theta;1)
	-
	\frac{1}{\sigma^2}
	-
	\frac{1}{\sigma^3} \mathbb{E}[F_3(X)]
	&
	D_{\sigma,\delta}(\Theta;1)
	\\[0,3cm]
	D_{\mu,\delta}(\Theta;1)
	&
	D_{\sigma,\delta}(\Theta;1)
	&
	D_{\delta,\delta}(\Theta;1)
	+
	2\mathbb{E}[F_4(X)] \
	\end{bmatrix}
},
\end{align*}
where 	$D_{u,v}(\Theta;n)$ is as in\eqref{D-def}.
We emphasize that, by using Theorem \ref{mgf-1} and Corollary \ref{mgf}, closed expressions for the expectations $\mathbb{E}[F_1(X)]$, $\mathbb{E}[F_2(X)]$ and $\mathbb{E}[F_3(X)]$ can be found. Furthermore, by using Theorem \ref{car-moments}, a simple argument shows that $F_4(X)$ is an integrable random variable.

\section{Real-world data analysis} \label{Sec:6}
\noindent

In this section, we used the BG model to fit one data set coming from  daily  pressure observations of the state of Mato Grosso do Sul-Brazil,
during the period of the year 2015 to 2020. The data series employed were taken from Corumbá automatic station of INMET (Instituto Nacional de Meteorologia) at https://portal.inmet.gov.br/dadoshistoricos. The required numerical evaluations are
implemented using the R software \cite{R}, \cite{hastings}, and \cite{MCMCpack}. The maximum likelihood estimates of the parameters were obtained  as in Section \ref{Sec:5}, using the bgumbel package \cite{Brom}.

For a random sample of the pressure $\{X_i\}_{t=1}^{T}$, $T=1774$, we fit their maximum values by
Gumbel and BG models.
For this, we used the maximum block technique. That is,
we provide $\tau=29$ non-overlapping subsamples of length $N=60$; $\{m_i\}_{i=1}^{\tau}$, where $m_i=\max\{x_{(i-1)N+1} \dots x_{i N + 1}\}$, \ $\forall i=1, \ldots, \tau=[T/N]$ ($[.]$ denotes the integer  part ). We applied Ljung – Box’s test \cite{LBox} for verifying the null hypothesis of serial independence of the new sample of  maxima.
The test statistics did not reject the null hypothesis at the significance level of $1,7\%$ for the $28$ blocks of size $N=60$ and  the last of size $34$.

The descriptive statistics of  sample $\{m_i\}_{i=1}^{\tau}$ and the sample centered by the  mean are shown in Table \ref{descriptive}.
\begin{table}[!htbp]
	\centering
	\caption{ Descriptive statistics}
	%\centering
	\begin{tabular}{cccccc}
		\hline
		Distribution & Mean & Median & Maximum & Minimum & Std Dev \\
		
		\hline
		
		Extreme values & 1009.2 & 1010.0 & 1017.9 & 1001.1 & 4.9301 \\
		
		Centralized extreme values & 0.0000 & 0.7893 & 8.7393 & $-$8.0607 & 4.9301 \\
		\hline
	\end{tabular}
	\label{descriptive}
\end{table}
%The graphical results of the adjustment of the data by the Gumbel distribution is %shown in the Figure. It is evident that the fit was not good.
%\cite{extRemes}
%\begin{figure}[!htbp]
%	\centering
%	\includegraphics[width=1\linewidth]{gumbel.png}
%	\vspace{-1cm}
%	\caption{Empirical  density (line) versus fitted Gumbel pdf  (dashed) and its QQ-plot. }
%	\label{fitdata}
%\end{figure}
The parameter estimates and their corresponding standard error estimates (SE) from the BG and Gumbel models are shown in Table \ref{estimates}.
The results of goodness-of-fit tests based on the Kolmogorov–Smirnov (KS) test, Akaike information criterion (AIC) and Bayesian information criterion (BIC)  indicate that the BG model is a better fit, judging on the basis of $p$-values and  the lowest value of the AIC and BIC, see Table \ref{estimates}. The Figures \ref{fitdata} and \ref{qqdata} also indicates that the adjustment of the data by the BG model is better than the simple Gumbel model.

\begin{table}[!htbp]
	\centering
	\caption{Parameters estimates of Gumbel and BG models.}
	%\centering
	\begin{tabular}{ccrcccc}
		\hline
		
		Model & Parameters & Estimate & SE & ks $p$-value & AIC & BIC\\
		
		\hline
		
		& $\mu$ & $-$3.6972 & 0.4613 & \\
		
		BG & $\sigma$ & 2.4661 & 0.2222 & 0.8239 & 169.8228 & 173.8194\\
		
		& $\delta$ & $-$0.4128 & 0.1220 & \\ \\
		
		Gumbel & $\mu$ & $-$2.4080 & 0.8667 & 0.6576 & 173.2287 & 175.8932\\
		
		& $\beta$ & 4.3326 & 0.6405 & 		
		\\
		\hline
	\end{tabular}
		\label{estimates}
\end{table}

\begin{figure}[!htbp]
	\centering
	\includegraphics[width=0.8\linewidth]{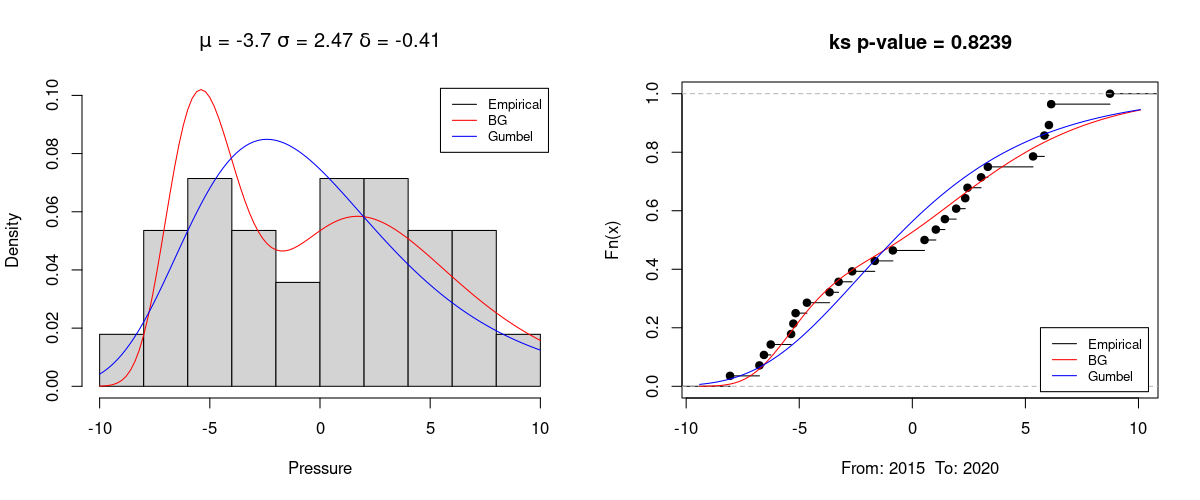}
%	\vspace{-1cm}
	\caption{On the left: centralized extreme values histogram versus fitted BG and Gumbel models. On the right: empirical cdf versus theorical cdf of BG and Gumbel models.}
	\label{fitdata}
\end{figure}

\begin{figure}[!htbp]
	\centering
	\includegraphics[width=0.8\linewidth]{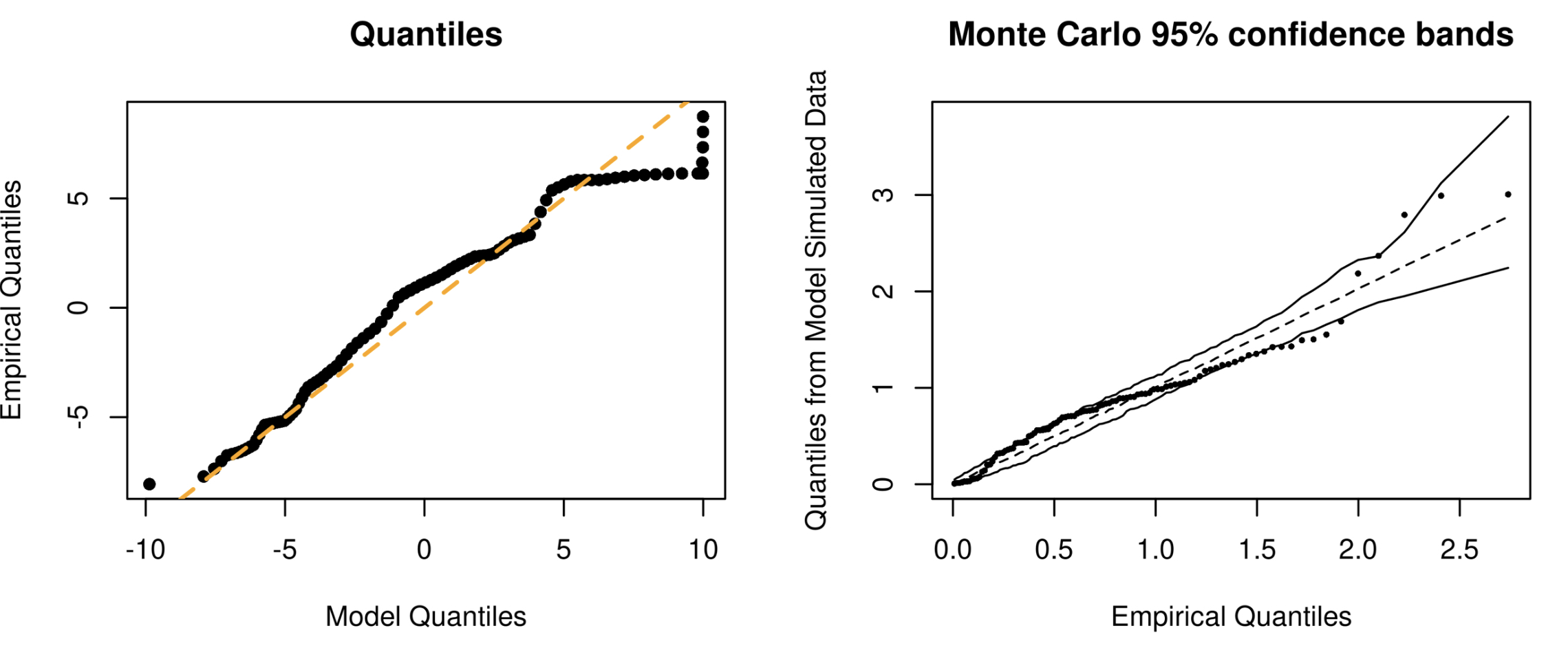}
	%\includegraphics[width=0.4\linewidth]{qqvv}
%	\vspace{-1cm}
	\caption{  QQ-plot  empirical versus  theoretical of  BG model. }
	\label{qqdata}
\end{figure}

\section{Conclusions} \label{Sec:7}
\noindent

In this paper, we proposed an extension to the Gumbel distribution  by using a quadratic transformation technique used to generate bimodal functions
produced due to using the quadratic expression, with an additional bimodality parameter
which modifies the mode of the distribution, composing as a alternative models for single
maxima events. In this generalization, the Gumbel distribution appears as a particular case.
We provide a mathematical treatment of the new
model including the mode, bimodality, moment generating function and moments.
We performed
the modelling under a frequentist approach and the estimation of the parameters was proposed
using the maximum likelihood estimation.
Finally, we have performed a statistical modeling with real data by using
the new proposed model in the article.
The application demonstrated the practical relevance of the new model, which also showed the advantage of Gumbel model.

\paragraph{Acknowledgements}
This study was financed in part by the Coordenação de Aperfeiçoamento de Pessoal de Nível Superior - Brasil (CAPES) (Finance Code 001).

\paragraph{Disclosure statement}
There are no conflicts of interest to disclose.

%\paragraph{ORCID}
%\noindent
%\\
%Cira E. G. Otiniano \
%\url{https://orcid.org/0000-0002-5619-0478}
%\\
%Roberto Vila \
%\url{https://orcid.org/0000-0003-1073-0114}
%\\
%Pedro C.  Brom \
%\url{https://orcid.org/0000-0002-1288-7695}
%\\
%Marcelo Bourguignon \
%\url{https://orcid.org/0000-0002-1182-5193}

%\newpage
%%%%%%%%%%%%%%%%%%%%%%%%%%%%%%%%%%%%%%%%%%%%%%%%%%%%%%%%%%%%%

\newpage

\section*{Appendix}

\begin{figure}[h!]
	\centering
	\includegraphics[width=0.8\linewidth]{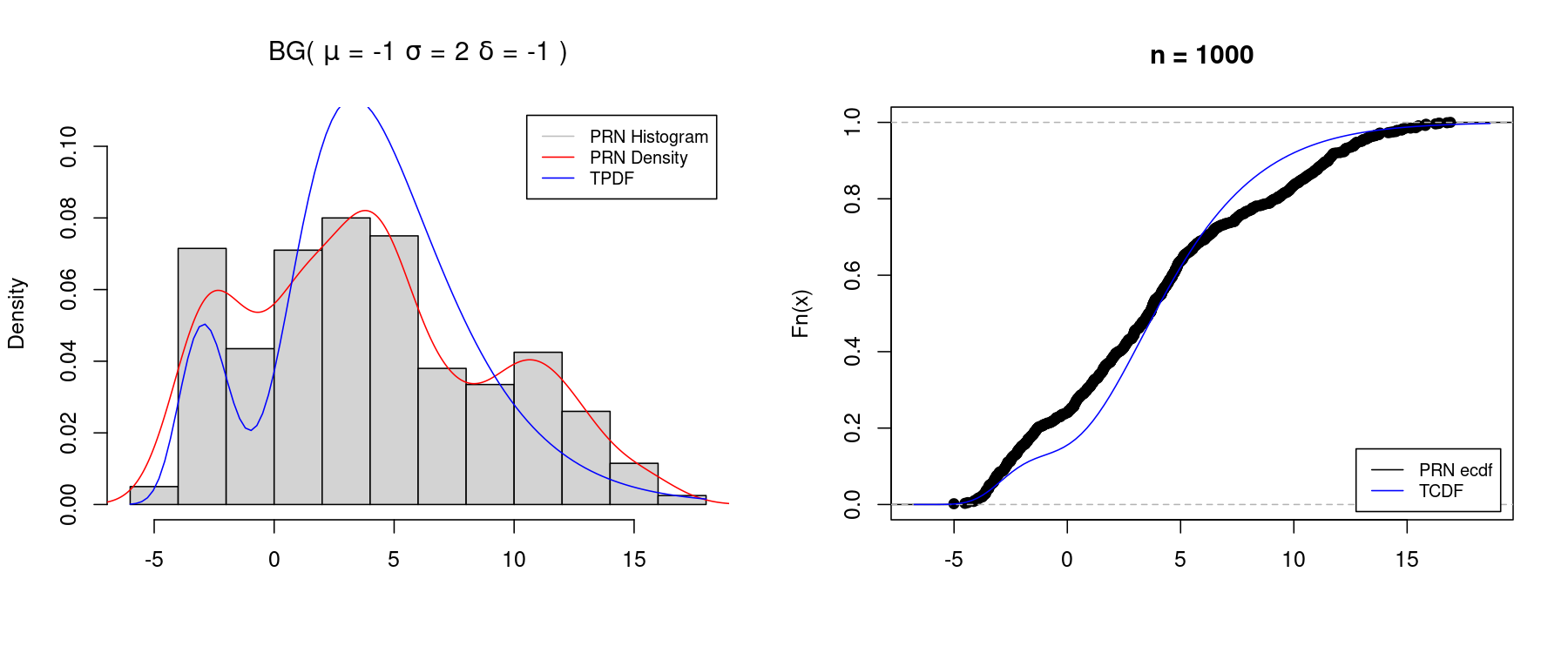}
	\vspace{-0.5cm}
	\caption{TPDF (left) and  TCDF (right),  n = 1000.}
	\label{figa1}
\end{figure}

\begin{figure}[h!]
	\centering
	\includegraphics[width=0.8\linewidth]{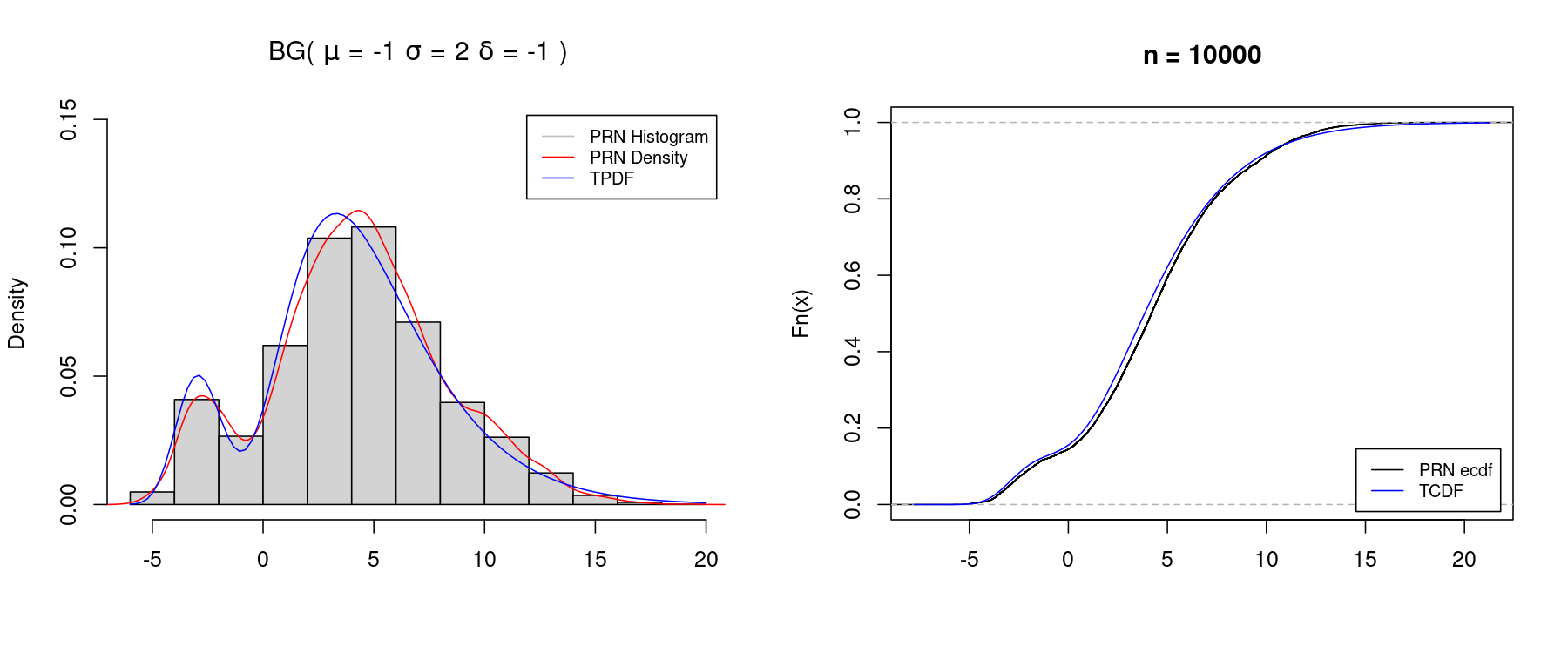}
	\vspace{-0.5cm}
	\caption{TPDF (left) and  TCDF (right),   n = 10000.}
	%\label{figa2}
\end{figure}

\begin{figure}[h!]
	\centering
	\includegraphics[width=0.8\linewidth]{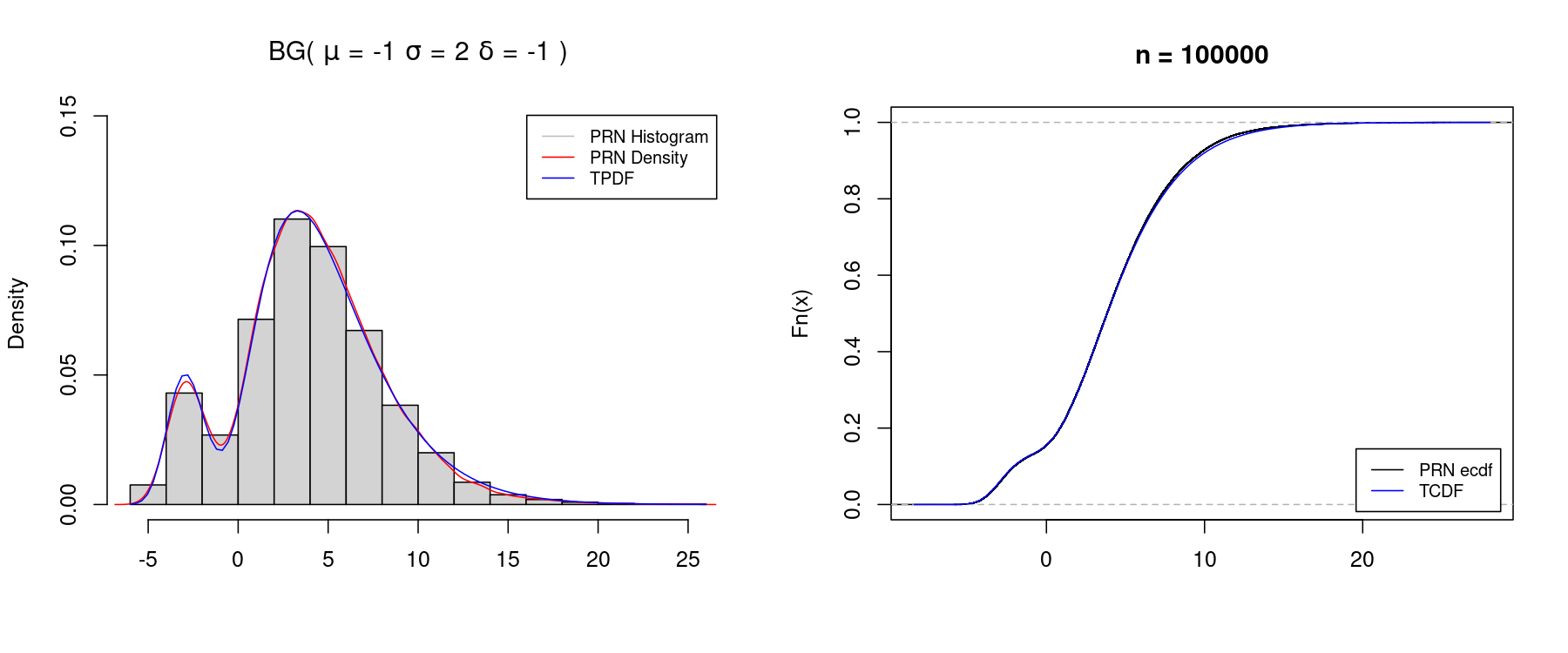}
	\vspace{-0.5cm}
	\caption{TPDF (left) and  TCDF (right),  n = 100000.}
	%\label{figa3}
\end{figure}

%%%%%%%%%%

\begin{figure}[h!]
	\centering
	\includegraphics[width=0.8\linewidth]{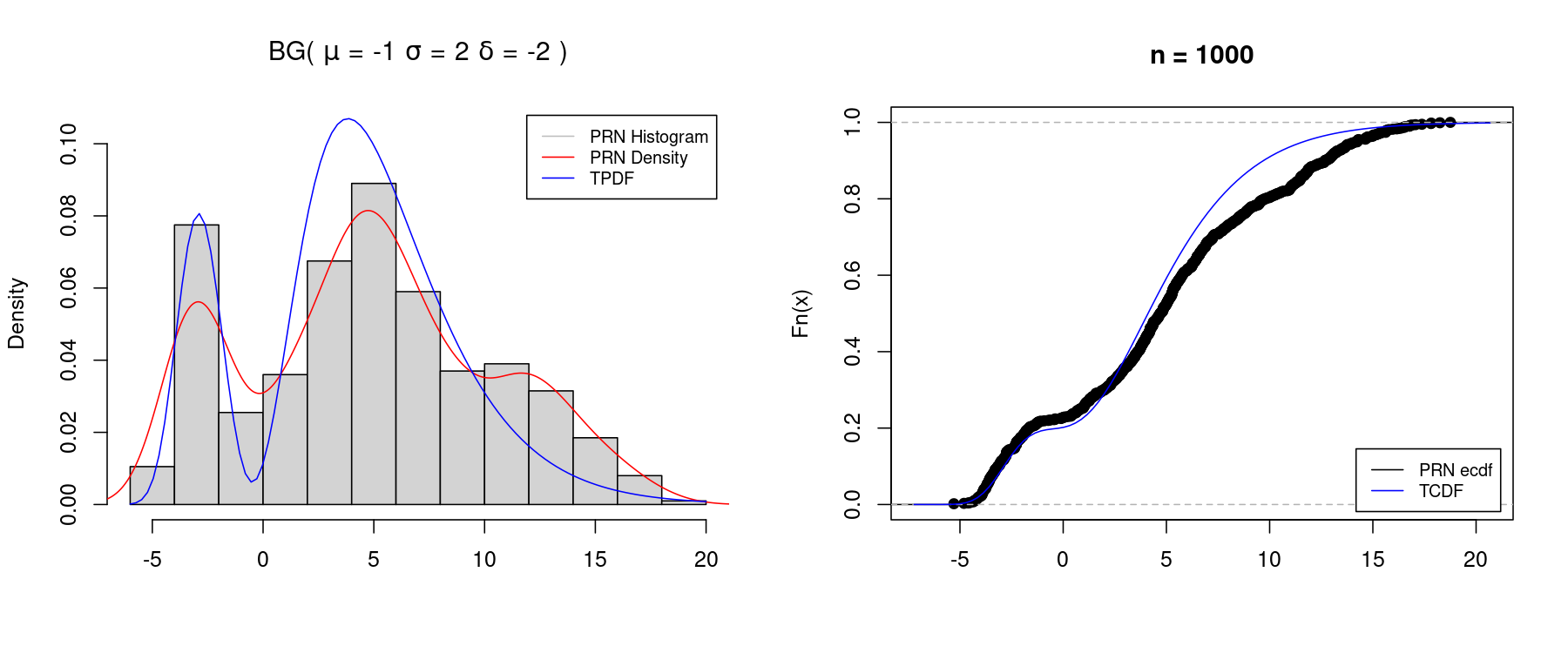}
	\vspace{-0.5cm}
	\caption{TPDF (left) and  TCDF (right),  n = 1000.}
	%\label{figa4}
\end{figure}

\begin{figure}[h!]
	\centering
	\includegraphics[width=0.8\linewidth]{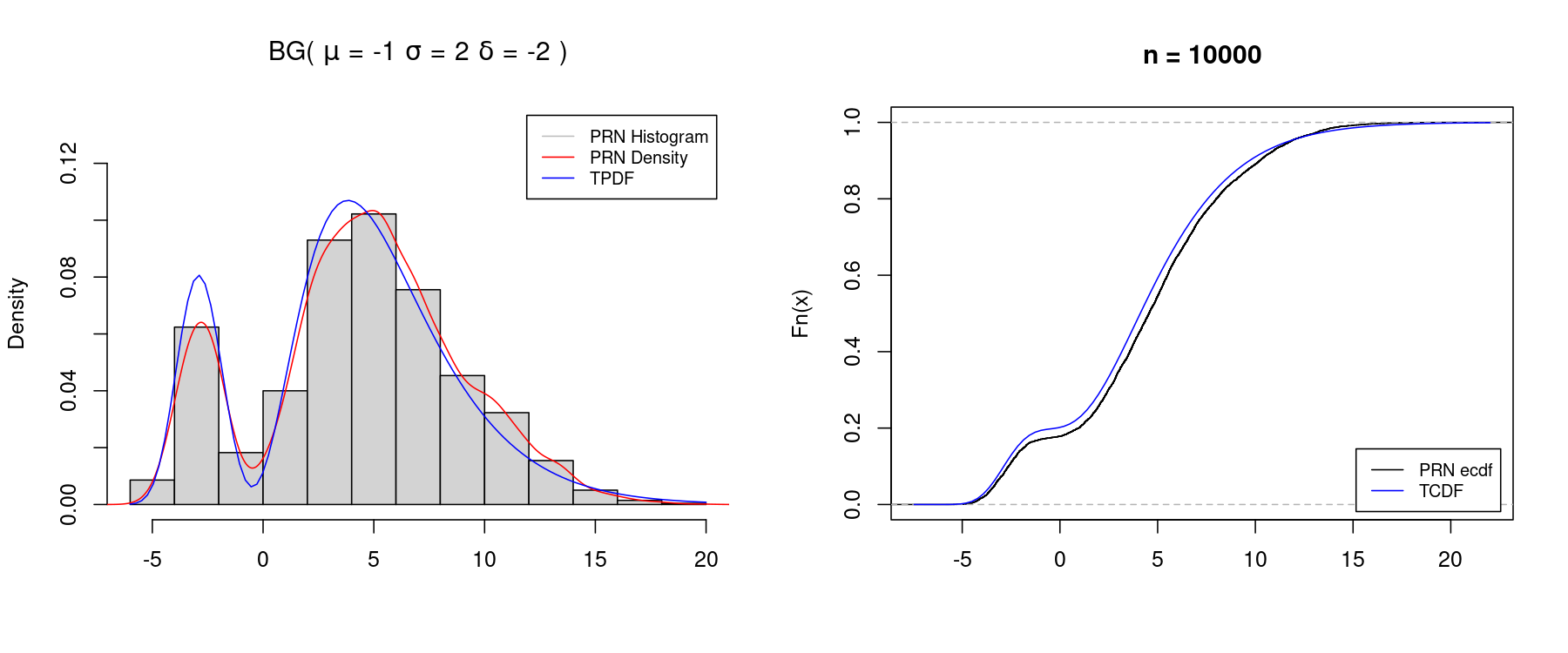}
	\vspace{-0.5cm}
	\caption{TPDF (left) and  TCDF (right),  n = 10000.}
	%\label{figa5}
\end{figure}

\begin{figure}[h!]
	\centering
	\includegraphics[width=0.8\linewidth]{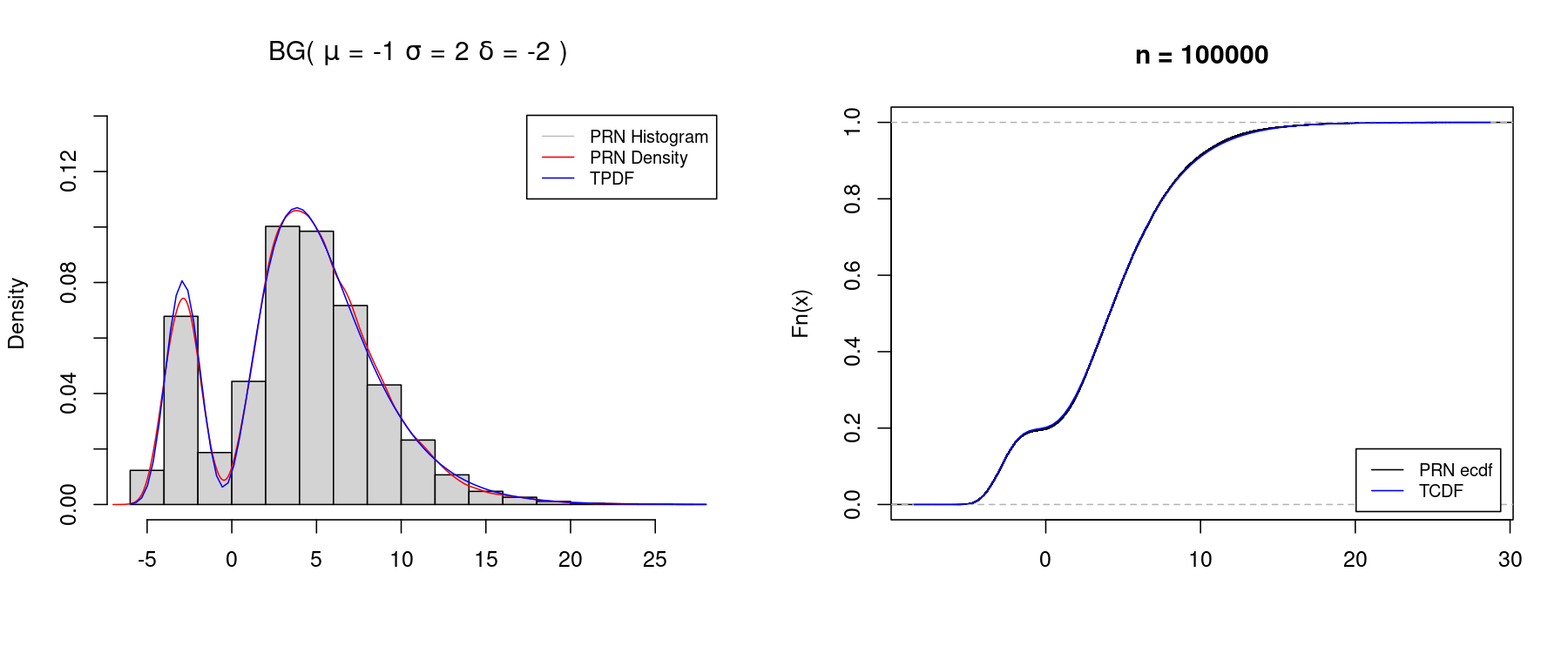}
	\vspace{-0.5cm}
	\caption{TPDF (left) and  TCDF (right),  n = 100000.}
	%\label{figa6}
\end{figure}

\begin{figure}[h!]
	\centering
	\includegraphics[width=0.8\linewidth]{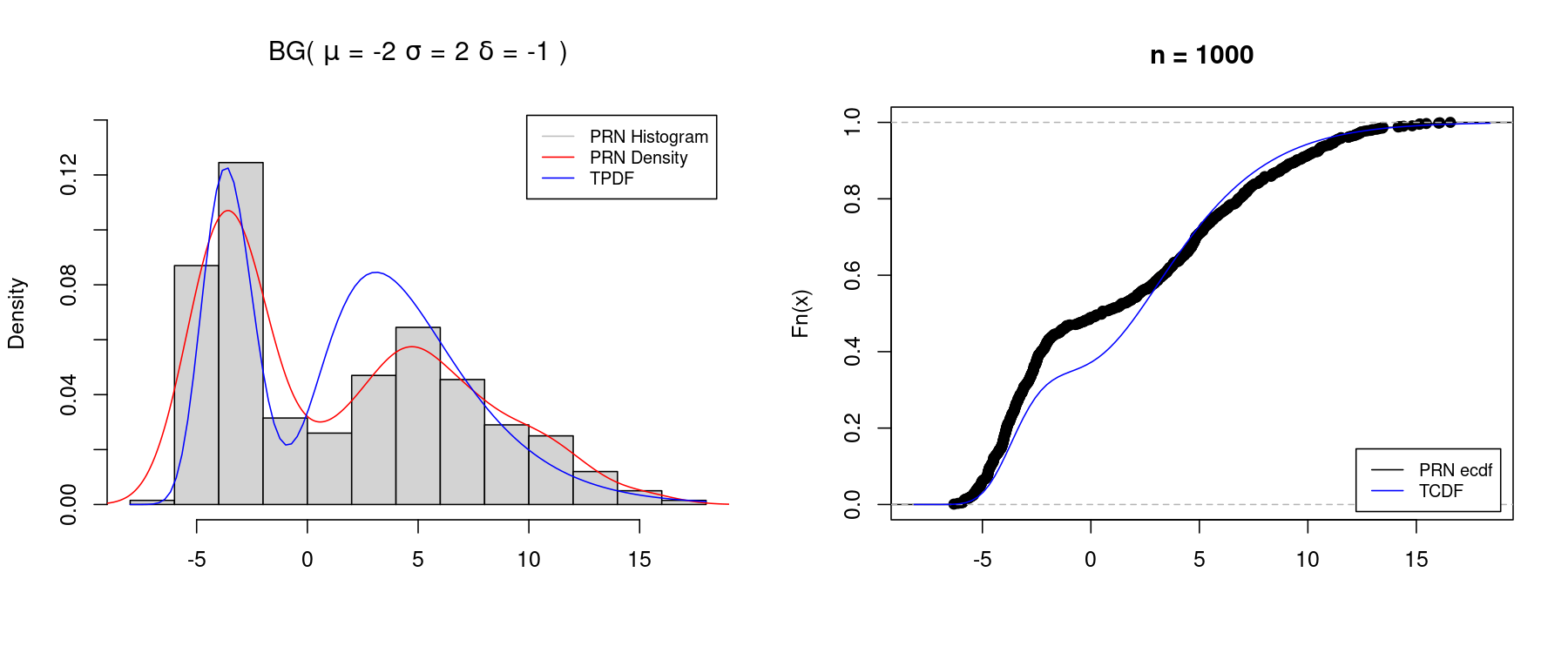}
	\vspace{-0.5cm}
	\caption{TPDF (left) and  TCDF (right),  n = 1000.}
	%\label{figa7}
\end{figure}

\begin{figure}[h!]
	\centering
	\includegraphics[width=0.8\linewidth]{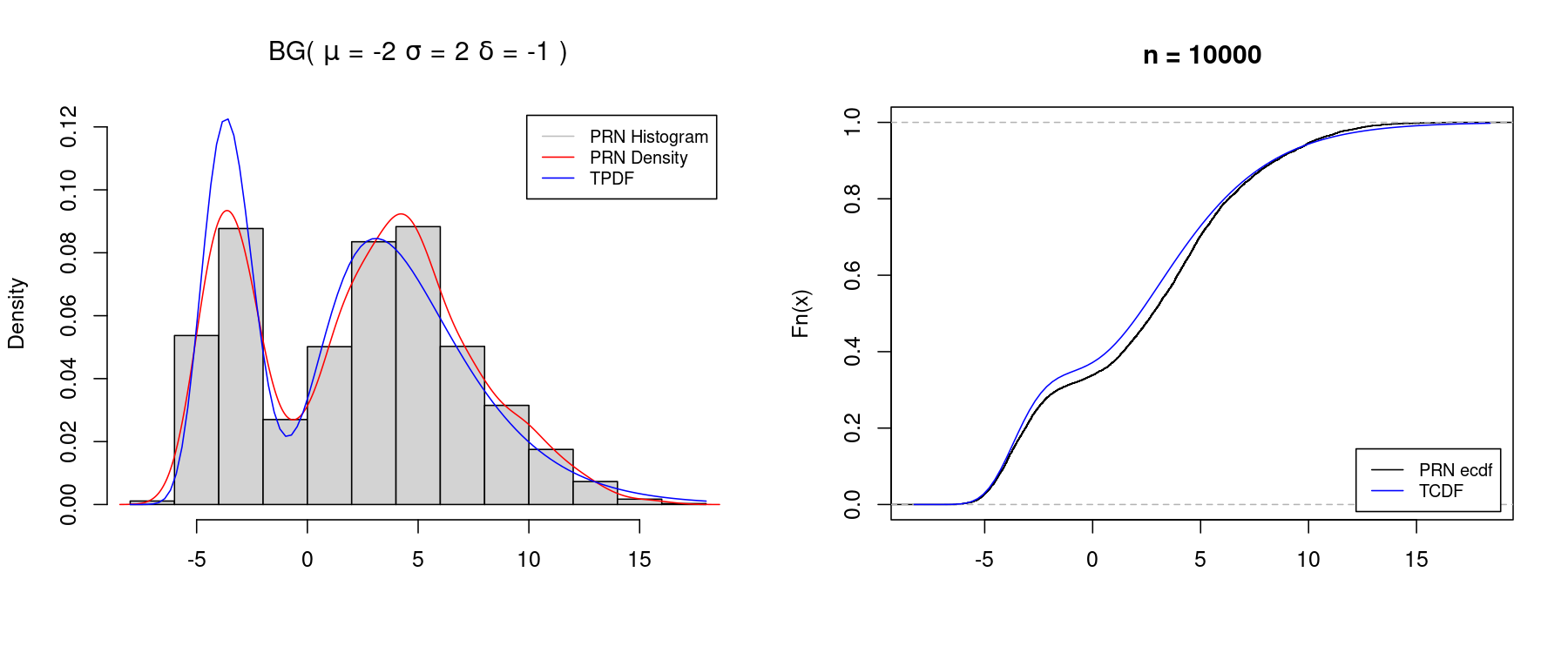}
	\vspace{-0.5cm}
	\caption{TPDF (left) and  TCDF (right),  n = 10000.}
	%\label{figa8}
\end{figure}

\begin{figure}[h!]
	\centering
	\includegraphics[width=0.8\linewidth]{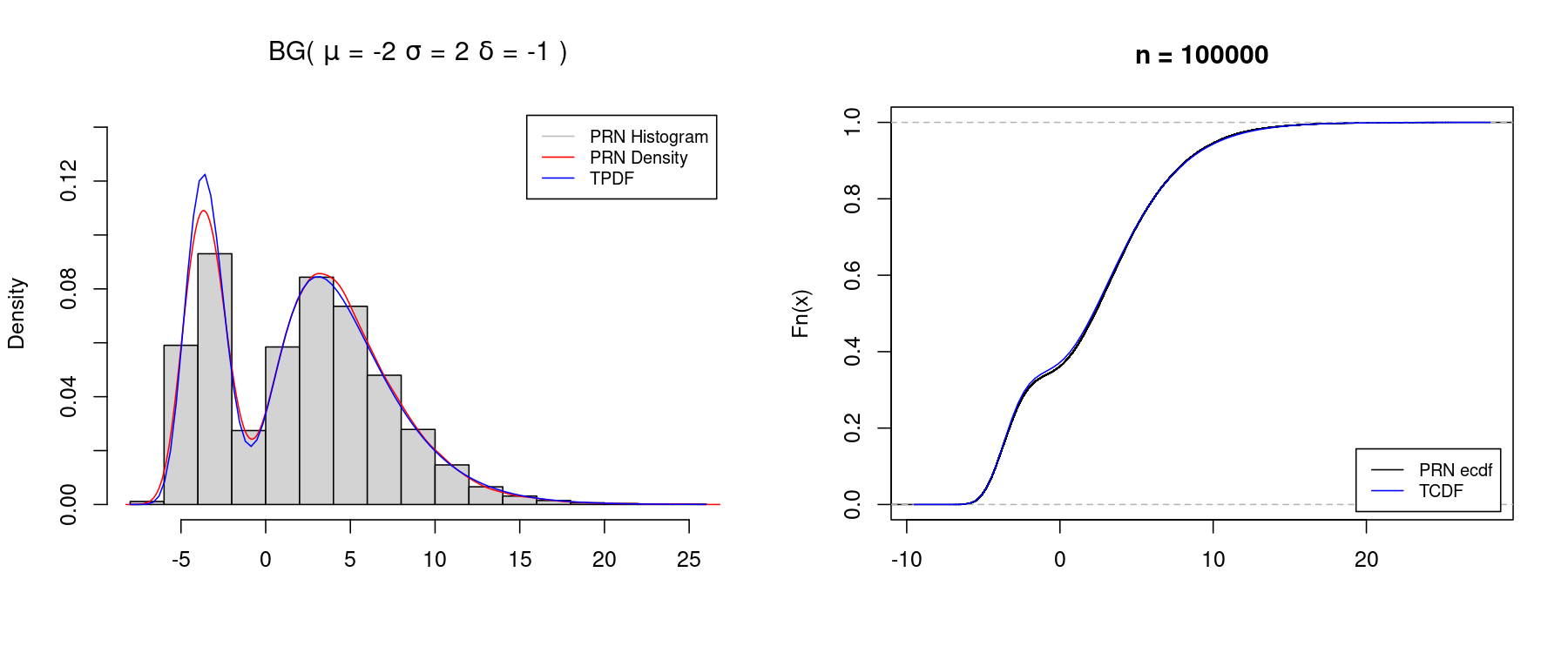}
	\vspace{-0.5cm}
	\caption{TPDF (left) and  TCDF (right),  n = 100000.}
	\label{figa9}
\end{figure}

\end{document}